\documentclass{article}

\usepackage[utf8]{inputenc}

\usepackage{color}
\usepackage{graphicx}
\usepackage{amsmath}
\usepackage{amsthm}
\usepackage{amssymb}
\usepackage{amsfonts}
\usepackage{xcolor}
\usepackage{caption}
\usepackage{mathrsfs}
\usepackage[T1]{fontenc}
\usepackage{cancel}
\usepackage{comment}
\usepackage{tikz}
\usepackage{floatrow}
\usepackage{mathtools}
\usepackage{hyperref}
\usepackage[labelformat=simple]{subcaption}

\usepackage{stackrel}

\usepackage{todonotes}
\usepackage{pdflscape}

\definecolor{FG}{rgb}{0.0, 0.5, 0.0}
	
\usepackage{amsmath,amsthm,amssymb,setspace,enumitem,epsfig,titlesec,verbatim,color,array,eurosym,multirow,blkarray,lscape,bm,nicefrac,comment,soul,graphicx,MnSymbol,wasysym,hyperref,multirow}
\usepackage[top=2.5cm,left=2.8cm,right=2.8cm,bottom=3.4cm]{geometry}
\usepackage{ulem}

\usetikzlibrary{decorations.pathreplacing,angles,quotes}

\usetikzlibrary{decorations.pathreplacing,angles,quotes}

\newtheorem{theorem}{Theorem}
\numberwithin{theorem}{section}
\newtheorem{corollary}[theorem]{Corollary}
\newtheorem{lemma}[theorem]{Lemma}
\newtheorem{proposition}[theorem]{Proposition}

\newtheorem{remark}[theorem]{Remark}
\definecolor{G1}{rgb}{0.0, 0.5, 0.0}

\allowdisplaybreaks

\title{Adding a fecundity-survival trade-off to a discrete population model with maturation delay}

\newcommand{\equalcontrib}{All authors contributed equally to this study.}

\author{Christopher J. Greyson-Gaito\thanks{\equalcontrib} \footnote{McMaster University, ON L8S 4L8, Canada, ORCID: 0000-0001-8716-0290, christopher@greyson-gaito.com}, Sabrina H. Streipert\footnotemark[1] \footnote{University of Pittsburgh, PA 15260, USA, ORCID: 0000-0002-5380-8818, streipert@pitt.edu}, Gail S.K. Wolkowicz\footnotemark[1] \footnote{McMaster University, ON L8S 4L8, Canada, ORCID: 0000-0002-4501-2342, wolkowic@mcmaster.ca}}

\date{\today}

\setstcolor{blue}

\begin{document}

\maketitle


\begin{abstract}
Although maturation delays are frequently included in population models, researchers rarely account for mortality between birth and maturity. Previous discrete population models have included mortality of immature individuals during the maturation delay finding that increasing the delay decreases the equilibrium population size, eventually leading to extinction. Since maturation delays beyond one breeding cycle are often found in nature, they must also have a benefit leading to a trade-off.  We derive a class of models to explore the trade-off between the benefit  of a longer maturation delay on fecundity due to  larger body sizes at maturity  and the down-side on survival. We examine two scenarios: density independent survival and cohort density dependent survival of immature individuals. For the mature and immature individuals, we consider two different, but popular, survival functions: the Beverton--Holt model and the Ricker model. Across all models, we identify a positive maturation delay that maximizes the population size that we refer to as  the ``optimal maturation delay'' and a critical delay threshold that results in extinction. We also find oscillatory dynamics with the Ricker survival function for certain ranges of maturation delay. Overall, our delay model sets up a useful phenomenological framework to test multiple combinations of trade-offs in parent survival, offspring survival, and reproductive investment.
\end{abstract}


\section{Introduction} 

Difference equations are an important and popular modeling framework to explore population dynamics. Reasons include the ease of modeling non-overlapping generations and processes that are discrete in nature such as breeding events, as well as, computational conveniences and the fitting of collected data that is discrete. The simplest difference population models that nevertheless can exhibit complex dynamics take the form $N_{t+1}=F(N_t)$, where $F\colon \mathbb{R} \to \mathbb{R}^+_0=[0,\infty)$ describes the net change in population size  of species $N$ over each breeding cycle. Examples of such models are the well-known Beverton--Holt and Ricker population models. In general, these models assume implicitly that newborn individuals contribute immediately to the next breeding cycle, therefore modeling individuals that reach maturity within one breeding cycle, where the time-intervals are chosen to replicate the length of a breeding cycle. Albeit mathematically convenient, this model assumption is rarely satisfied in nature \cite{buddUniversalPredictionVertebrate}.

One extension of these simple discrete population models is to make the growth and loss rates  age dependent. In that case, an age-structured difference equation model would take the form $\overrightarrow{N}_{t+1}=\overrightarrow{F}(\overrightarrow{N}_t)$, where $\overrightarrow{N}=(N^{(1)},\ldots,N^{(A)})^T\ \in \mathbb{R}^A$ and $N^{(i)}$ represents the $i$-th age class of  species $N$ (see for example \cite{wikanAgeStageStructure2012,cushingIntroduction1998}). However, such models require the collection of   age-dependent data  that is  not always available. Alternatively, as a compromise between the simple model and an age-structured population model,  delay recurrences of the form $N_{t+1}=F(N_t,N_{t-\tau})$ effectively blend the simple scalar form and the vector-valued form \cite{Clark1976, derisoHarvestingStrategiesParameter1980}. Commonly, the delay is either incorporated in the fitness function, resulting in the form  $N_{t+1}=N_tf(N_{t-\tau})$  \cite{liChaosDiscreteDelay2012,sakerOscillationAttractivityDiscrete2007}, or in the reproduction function $N_{t+1}=f_1(N_t)+f_2(N_{t-\tau})$ \cite{Clark1976, panStabilityNeimarkSacker2021,streipertAlternativeDelayedPopulation2021}. The latter assumes that the mature subpopulation at time $t+1$ depends on those individuals that were mature at time $t$ and survived to time $t+1$, modeled by $f_1(N_t)$, and those individuals that were born at time $t-\tau$ and reach maturity at time $t+1$, modeled by $f_2(N_{t-\tau})$. Most single species models, where the delay has been incorporated in the reproduction, have not accounted for the losses that occur between birth and maturity (e.g., \cite{Clark1976,el-morshedyGloballyAttractingFixed2006}).

To address the problem of not accounting for immature individual losses, the discrete population model with delay in the reproduction, introduced in \cite{StWodelay2}, included an immature survival probability. This model describes the mature individuals of a species that reach maturity after several breeding cycles (time units) and has the form:
\begin{equation}\label{eq:SWdelaymodel}
    N_{t+1}=p(N_t)N_t + g(N_{t-\tau})\widetilde{p}(\tau,N_{t-\tau})N_{t-\tau},
\end{equation}
where $p(N_t)$ describes the fraction of already mature individuals $N_t$ that survive to time $t+1$, $g(N_{t-\tau}){N_{t-\tau}}$ models the number of newborns produced right after time $t-\tau$, determined by the mature individuals at time $t-\tau$, and $\widetilde{p}(\tau,N_{t-\tau})$ denotes the fraction of immature individuals that were born at time $t-\tau$ and survived to time $t+1$, when they reach maturity for the first time. Later, in \cite{StWodelay3}, \eqref{eq:SWdelaymodel} was extended by considering a distributed delay with a maturation kernel, $K(s)$, that describes the fraction of individuals that mature $s+1$ time units after they are born. If $K(s)=1$ for $s=\tau\in \mathbb{N}$ and $K(s)=0$ for all $s\neq \tau$, then this model is consistent with \eqref{eq:SWdelaymodel}. Notably, both models, \cite{StWodelay2,StWodelay3} show that, under realistic assumptions, there exists a critical delay threshold, $\tau_c$, such that if the (earliest) maturation exceeds $\tau_c$, then the species goes extinct. Therefore, these studies suggest that natural selection should favour early maturation times, i.e., reducing $\tau$ to avoid extinction. 

However, maturation delays are common across species and can range from 0.06 years for the California vole (\textit{Microtus californicus}) to 25.52 years for the Aldabra giant tortoise (\textit{Alabrachelys gigantea}) \cite{buddUniversalPredictionVertebrate}. Furthermore, there are several examples of variation in maturation delay within a species including earlier maturation of the spotted skink (\textit{Carinascincus ocellatus}) in warmer versus colder sites \cite{augertPlasticityAgeMaturity1993}, and heritable variation of maturation delay in the Atlantic Salmon (\textit{Salmo salar}) \cite{raunsgardVariationPhenotypicPlasticity2024}. Considering that maturation delays are ubiquitous, if all else is equal, increasing a maturation delay leads to extinction, \eqref{eq:SWdelaymodel} must be missing another key aspect. Indeed, the maturation delay (sometimes called {\it age-at-maturity} or {\it size-at-maturity} in the biological literature) is a critical life history trait. In any life history trait, there is an energy allocation trade-off between maintenance, growth, survival, and reproduction, often leading to a benefit that is  traded-off with a negative consequence \cite{zeraPhysiologyLifeHistory2001a}. For example, a shorter larval period in \textit{Drosophila melanogaster} leads to reduced larval viability and diminished adult size \cite{chippindaleExperimentalEvolutionAccelerated1997}. 
In contrast, a longer larval period leads to increased accumulation of lipids and increased conservation of larval reserves in young adults \cite{chippindaleExperimentalEvolutionAccelerated1997}. Consequently, the models in \cite{StWodelay2, StWodelay3}, while incorporating the negative aspect of increased total loss of immature individuals prior to reaching maturity, do not include any benefits of  longer maturation delays.

One benefit of longer maturation delays can be higher fecundity. Generally, fecundity increases with size which is usually postively correlated with age \cite{roff2002life}. In animals specifically, there is usually a threshold of size that must be reached for reproduction to occur \cite{reesGrowthReproductionPopulation1989}. For example, \cite{winemillerPatternsLifehistoryDiversification1992} found a positive association between adult body size with delayed maturation and clutch size across 216 North American fish species. Another example is the green sea turtle (\textit{Chelonia mydas}) that completes most of its growth prior to maturity with a positive relationship between the average number of eggs in a nest and female body size \cite{Phillips2021, Frazer86, Carr89}. If we assume, as commonly done \cite{Spence2017, vonB2}, that size and age are linearly related, then we get the benefit of an increase in fecundity with longer maturation delays. Adding higher fecundity with longer maturation delays is then a promising addition to the above mentioned delayed population models.

Motivated by this relationship between  age, length, and increased fecundity, we start with an attempt at incorporating a trade-off for the maturation delay. Specifically, we extend the discrete population model in \cite{StWodelay2} to include a trade-off, where earlier maturation leads to lower fecundity. We compare the impact of Beverton--Holt and Ricker survival functions and consider, for each of these mature individual survival functions, two different scenarios of survival for immature individuals: i) density-independent survival, and ii) cohort-dependent survival. Overall, we ask how a longer maturation delay that increases total immature losses interacts with its incorporated trade-off of higher fecundity to affect the population dynamics. We find that the trade-off allows for a higher maturation delay compared to no trade-off. We also find within the Ricker framework that a longer maturation delay allows for the existence of oscillatory population levels, prior to extinction. 

\section{Trade-off Model Formulation} 

Our model formulation follows the set-up of \cite{StWodelay2}, where the survival of mature and immature individuals are separated, to allow for an introduction of the delay solely in the reproductive term. More precisely, we consider
\begin{equation}\label{eq:delay2}
    N_{t+1}=p(N_t)N_t + g(N_{t-\tau})\tilde{p}(\tau,N_{t-\tau})N_{t-\tau},
\end{equation}
where $p(N_t)$ describes the fraction of already mature individuals $N_t$ that survive to time $t+1$, $g(N_{t-\tau})$ is the growth rate at time $t-\tau$, determined by the mature individuals at time $t-\tau$, and $\tilde{p}(\tau,N_{t-\tau})$ is the fraction of immature individuals that were born at time $t-\tau$ and survive to time $t+1$, where they reach maturity for the first time. In \cite{StWodelay2}, the authors obtained general results for the dynamics of \eqref{eq:delay2} and formulated, as examples, a delayed Beverton--Holt and a delayed Ricker model. 
For these two specific examples, we consider a constant growth rate, $g(N)\equiv g$, which is consistent with the classical formulations without delay. Instead, in this work, to incorporate a trade-off that involves  a positive effect of a delayed maturation,  $g(N)$ in \eqref{eq:delay2} is replaced by $g(\tau, N)$.

We therefore consider a refinement of \eqref{eq:delay2} that incorporates the positive effect of delay on the growth rate (e.g., number of eggs produced), 

\begin{equation}\label{delay2star}
   N_{t+1}=p(N_t)N_t 
+ \overbrace{g(\tau, N_{t-\tau})\widetilde{p}(\tau,N_{t-\tau})}^{=:m(\tau, N_{t-\tau})}N_{t-\tau},   \quad N_i=N_i^0 \geq 0, \, \, i  \in \{-\tau, -\tau-1, \ldots, 0\}. 
\end{equation}


We assume that $p\in C^1(\mathbb{R}^+_0, [0,1])$,  $\widetilde{p}\in C^1(\mathbb{N}_0\times \mathbb{R}^+_0, [0,1])$, $g\in C^1(\mathbb{N}_0\times \mathbb{R}^+_0, \mathbb{R}^+)$, and $\tau\in \mathbb{N}_0$ represents the maturation delay.   We assume the function  $g(\tau, x)$ is increasing in $\tau$ to account for a benefit in delayed maturation as described below.

We will often require the following assumption
\begin{equation}\label{eq:H}(H) \qquad p(z) \mbox{ is strictly decreasing for } \, z> 0 \mbox{ and }  \lim_{z\to \infty}p(z)=0.
\end{equation}

We address some basic results for this general expression \eqref{delay2star}. The first, regarding the nonnegativity of solutions, immediately follows from the structure of the recurrence.

\begin{proposition}\label{Proppos}
Consider \eqref{delay2star}, for fixed $\tau\in \mathbb{N}$. If $N_i^0=0$ for all $i\in \{-\tau, \ldots, 0\}$, then $N_t=0$ for all $t\geq 0$. Furthermore, if there is at least one $j\in \{-\tau, \ldots, 0\}$ such that $N_j^0>0$, then $N_t>0$ for all $t>\tau$. 
\end{proposition}

\begin{proposition}\label{Propbound}
    Consider \eqref{delay2star}, for fixed $\tau\in \mathbb{N}$.  If $\underline{m}:=\mbox{inf}_{x\geq 0}m(\tau, x)>1$, then the solution is unbounded if $N_j^0>0$ for at least one $j\in \{-\tau, \ldots, 0\}$.
\end{proposition}

\begin{proof}
Note that 
$$N_{t+1}=p(N_t)N_t + m(\tau, N_{t-\tau}) N_{t-\tau}>\underline{m} N_{t-\tau}.$$
This implies that $N_{t+1}>\underline{m}^{j\tau}N_{t-j\tau+1}$ so that  $N_t$ is unbounded.
\end{proof}


{
\begin{proposition}\label{Prop:bound}
     Consider \eqref{delay2star}, for fixed $\tau\in \mathbb{N}$. 
    Define $\overline{m}:=\max_{0\leq i\leq \tau}m(\tau,N_{-i})$. 
     Assume  that either ($p(0)+
     \overline{m}\leq 1$  and (H) defined in \eqref{eq:H} holds) or  ($p(0)+\overline{m} > 1$,  $\overline{m}<1$,  and  $\limsup_{x\to\infty} xp(x)$ is finite). 
     Then, all solutions with nonnegative initial  conditions  are bounded. 
 \end{proposition}

 \begin{proof}  
Let  $\overline{N}^0=\max_{0\leq i\leq \tau}\{N_{-i}\}$.
First, assume that  $p(0)+\overline{m}\leq 1$ and (H) holds. Then, $N_{t+1}=p(N_t)N_t+m(\tau, N_{t-\tau})N_{t-\tau}\leq(p(0)+m(\tau, N_{t-\tau}))\max\{N_t,N_{t-\tau}\} \leq \overline{N}^0$,     for all $t\in \mathbb{N}$,  and hence solutions are bounded by $\overline{N}^0$.  

 Next, assume that $p(0)+\overline{m}>1$ and $\limsup_{x\to\infty} xp(x)$ is finite. 
 Then, there exists $M>0$ such that $p(x)x \leq M$ for all $N \in \mathbb{N}$ and 
     $0\leq \overline{N}^0\leq  \frac{M}{1-\overline{m}}$. Then, 
     \begin{align*}
       N_1&= p(N_0)N_0 +m(\tau, N_{-\tau}) N_{-\tau}\leq p(N_0)N_0 +\overline{m}\overline{N}^0 
       \leq M +\overline{m}\frac{M}{1-\overline{m}}=
     \frac{M}{1-\overline{m}}.
 \end{align*}
Inductively, it follows that all solutions are   bounded above by $\frac{M}{1-\overline{m}}$ and by Proposition~\ref{Proppos} they are bounded below by zero.  
 \end{proof}





\subsection{Growth rate}
We   assume that fecundity increases with size, and once maturity is reached individuals do not grow any further, similar to \cite{stearnsEvolutionPhenotypicPlasticity1986}. 
First, we use the following length ($\ell$) and fecundity ($F$) relationship 
\begin{equation}\label{specform}
F(\ell)=A\cdot \, \ell - B,
\end{equation}
where $A, B>0$ are species-specific and typically obtained by fitting \eqref{specform} to available fecundity at length data \cite{Niu,biology12010131}.
Next, we link length to age of individuals using the von Bertalanffy curve \cite{vonB, LatA},
\begin{equation}\label{vonBerta}
\ell (a)=L_{\infty}\left(1-e^{-K(a-a_0)}\right),
\end{equation}
where $\ell(a)$ is the length at age $a$ and $L_\infty>0$ is the maximum length that can be reached by an (average) individual. Thus, the maximum number of eggs is produced if the individual reaches values close to $L_\infty$, but this may also require several years, dependent on the parameter $K, a_0>0$. 
Using the monotone relation to age in \eqref{vonBerta} and the power relation with egg production in \eqref{specform}, we link our maturation delay with the von Bertalanffy curve to obtain the function $g(\tau)$ that sets the average number of eggs a mature individual produces within a breeding cycle:
\begin{equation}\label{eq:growth}
g(\tau)=A L_\infty (1 -  e^{-K(\tau+1-a_0)})-B
=a - b e^{-K(\tau+1)}, \quad \mbox{for}\,\, a,b, K>0.
\end{equation}

To avoid negative or zero egg production even in the case with no maturation delay, we assume that $a>be^{-K}$. This allows for at least some offspring, albeit likely a small number, even in the case of no maturation delay. Note that, for fixed $\tau$, this construction of the growth rate $g$ is consistent with the assumption of a constant growth rate often found in population formulations such as the classical Beverton--Holt and Ricker population models. 
Thus, the models we consider are of the form
\begin{equation}\label{EqmainT1}
N_{t+1}=p(N_t)N_t+m(\tau,N_{t-\tau})N_{t-\tau}, \quad N_i=N_i^0 \geq 0, \, \, i  \in \{-\tau, -\tau-1, \ldots, 0\},
\end{equation}
where 
\begin{equation}\label{eq:mtaugen}
m(\tau,N_{t-\tau})=g(\tau)\widetilde{p}(\tau,N_{t-\tau})
\end{equation}
with $g(\tau)=a-be^{-K(\tau+1)}$ given in \eqref{eq:growth} and $a, b, K>0$, $p\in C^1(\mathbb{R}^+_0, [0,1])$, $\widetilde{p}\in C^1(\mathbb{N}_0\times \mathbb{R}^+_0, [0,1])$.

\subsection{Mature and immature survival}

For the survival function of already mature individuals, we use either the Beverton--Holt function,
\begin{equation} \label{eqn:BHmatsurvival}
    p(N_t)=\frac{1}{1+\alpha+\beta N_t},
\end{equation}
or the Ricker function,
\begin{equation} \label{eqn:Rmatsurvival}
    p(N_t)=e^{-\alpha-\beta N_t},
\end{equation}
with natural death parameter $\alpha>0$ and competition parameter $\beta>0$. The Beverton--Holt survival and Ricker survival are two commonly used forms in fishery sciences \cite{BH, Hilborn, Quinn, Ricker} and are well-studied discrete single species models (e.g.,  \cite{brauer2013mathematical, may2001stability,Keshet}).

For the survival of immature individuals, we consider two scenarios: i) constant survival of immature individuals, and ii) cohort density-dependent survival of immature individuals.

In scenario i), the constant survival of immature individuals,  we assume that the survival of immature individuals has only negligible density effects overall, so that the survival of an individual until it reaches maturity is given by
\begin{equation} \label{eqn:constantimmature}
    \widetilde{p}(\tau, N_{t-\tau})=\widetilde{p}(\tau)=\overline{p}^{\tau+1}
\end{equation}
for some $\overline{p}\in (0,1)$. For this scenario, we establish some general results, independent of the specific form of survival of mature individuals. However, to discuss the global dynamics, we consider two distinct population models of the form \eqref{EqmainT1}, one that considers the Beverton--Holt mature survival \eqref{eqn:BHmatsurvival} and the other one that considers the Ricker mature survival \eqref{eqn:Rmatsurvival}. We refer to these models as ``BH-Constant'' and ``Ricker-Constant''.

In our second scenario ii), we assume that the mature individuals have a negligible density-dependent effect on the survival of immature individuals (i.e. they occupy different ontogenetic niches), so that immature individuals are solely exposed to density-effects within their own cohort. Again, we distinguish between the popular forms of Beverton--Holt and Ricker survival. In the Beverton--Holt formulation of the cohort density-dependent survival of immature individuals, we consider
    $$\mbox{Prob}(\mbox{survival of immature individuals $w_t$ from $t$ to $t+1$}) =s(w_t)=\frac{1}{1+D+Cw_t}, \qquad D, C>0.$$

 Then, as in the derivation of earlier work \cite{StWodelay2}, we obtain the survival of immature individuals $\hat{w}$ over the entire maturation period of length $\tau+1$, as 
 $$\widetilde{p}(\tau, \hat{w})=\frac{D}{D(1+D)^{\tau+1}+((1+D)^{\tau+1}-1)C\hat{w}},$$
 where $\hat{w}$ are the immature individuals that are born at time $t-\tau$ that become mature, for the first time, at time $t+1$. Because $\hat{w}=g(\tau)N_{t-\tau}$, the immature survival function becomes
\begin{equation}\label{eq:BHimmsurvival}
\widetilde{p}(\tau, N_{t-\tau})=
\frac{D}{D(1+D)^{\tau+1}+((1+D)^{\tau+1}-1)Cg(\tau)N_{t-\tau}}.
\end{equation}

In the Ricker formulation of the cohort density-dependent survival of immature individuals, individuals of age $i$ at time $t$, that is of size $J_t^{(i)}$, have survival probability over one year of $e^{-D-CJ_t^{(i)}}$. Thus the survival probability of immature individuals over their maturation period is calculated from the full stage-structured model that takes the form
\begin{equation} \label{eqn:RickerRicker}
   \begin{bmatrix}
    N_{t+1}\\
    J_{t+1}^{(1)}\\
    J_{t+1}^{(2)}\\
    \ldots\\
    J_{t+1}^{(\tau)}\\
    \end{bmatrix} = \begin{bmatrix}
e^{-\alpha-\beta N_t} & 0 & 0 & \ldots & 0 & e^{-D-CJ_{t}^{(\tau)}}\\
 g(\tau)e^{-D-C g(\tau)N_{t}} & 0 & 0 & \ldots & 0 & 0\\
0 & e^{-D-CJ_{t}^{(1)}} & 0 & \ldots & 0 & 0\\
\ldots & \ldots & \ldots & \vdots\vdots\vdots & \ldots & \ldots\\
0 & 0 & 0 & \ldots & e^{-D-CJ_{t}^{(\tau-1)}} & 0\\
\end{bmatrix}
\begin{bmatrix}
    N_t\\
    J_t^{(1)}\\
    J_t^{(2)}\\
    \ldots\\
    J_t^{(\tau)}\\
\end{bmatrix},
\end{equation}
assuming that the natural death factor $D>0$ and density-dependence factor $C>0$ are age independent. These assumptions could however easily be relaxed. Note that in this set-up, the survival probability $\widetilde{p}(\tau, N_{t-\tau})$ has to be obtained computationally by evaluating the composition of survival functions. More precisely, for $f(x)=e^{-D-Cx}x$, $g(\tau)\widetilde{p}(\tau, N_{t-\tau})N_{t-\tau}=f\circ \ldots \circ f\circ f(g(\tau)N_{t-\tau})$. Therefore, the recurrence for the mature individuals at time $t+1$ is
\begin{equation*}
    N_{t+1}=e^{-\alpha-\beta N_t}N_t+f\circ \ldots \circ f\circ f(g(\tau)N_{t-\tau}),
\end{equation*}
for $\alpha, \beta, D, C>0$ and $g(\tau)$ given in \eqref{eq:growth} with $a,b>0$ and $a>be^{-K}$.

In our analysis of scenario ii), the modeling frameworks are always matched between mature and immature survival functions. That is, if we consider a Beverton--Holt survival for immature individuals, then the same applies for mature individuals. Similarly, if we assume that immature individuals follow a Ricker survival, then so do adults, see the first row in \eqref{eqn:RickerRicker}. We refer to the two models in scenario ii) as "BH-BH'' and "Ricker-Ricker''.

\section{Scenario i) \textemdash Constant survival of immature individuals}

Since we are assuming constant survival of immature individuals, in order for the population to avoid extinction and be  bounded, we make the following assumption in this section.
\begin{equation}\label{eq:A}
(A) \qquad \qquad \overline{N}^0:=\max_{i\in \{-\tau, \ldots, 0\}}\{N_i^0\}>0, \qquad \mbox{and}\qquad m(\tau)\in (0,1).
\end{equation}

We begin our analysis by considering \eqref{EqmainT1} in scenario i). That is, we assume that  over one time period the immature individuals have a constant survival probability, $\overline{p}\in (0,1)$ so that $\widetilde{p}(\tau, N_{t-\tau})=\overline{p}^{\tau+1}$. Substituting \eqref{eqn:constantimmature} into \eqref{EqmainT1}, results in a density-independent trade-off for immature individuals and yields
\begin{align}\label{EqT1C1}
& \qquad \qquad N_{t+1}=p(N_t)N_t + \left(a - b e^{-K(\tau+1)}\right)
\overline{p}^{\tau+1} N_{t-\tau}=p(N_t)N_t+m(\tau)N_{t-\tau}\nonumber \\
  & \mbox{with initial conditions}\\
 & \qquad \qquad N_i=N_i^0 \geq 0, \qquad \qquad i  \in \{-\tau, -\tau-1, \ldots, 0\}, \nonumber
\end{align}
where, by \eqref{eq:mtaugen}, 
\begin{equation}\label{eq:mtau}
m(\tau)=m(\tau, N_{t-\tau})=g(\tau)\widetilde{p}(\tau, N_{t-\tau})=\left(a - b e^{-K(\tau+1)}\right)
\overline{p}^{\tau+1}.
\end{equation}
All model parameters $a, b, K$, and  $\overline{p}\in (0,1)$ are positive.

Recall that we assume $a>be^{-K}$.

Note that $N^*_0=0$ is always an equilibrium of \eqref{EqT1C1}. A positive equilibrium $N^*_+$ exists if and only if 
$$1=p(N^*_+)+m(\tau) \qquad \iff \qquad 1-m(\tau)=p(N^*_+).$$
Since, by (H) given in \eqref{eq:H}, $p$ is strictly monotone decreasing and converges  to $0$ and $1-m(\tau)$ is fixed, for fixed $\tau\geq 0$, there exists at most one unique positive equilibrium $N^*_+$. This equilibrium $N^*_+$ only exists if $1-m(\tau)<p(0)$. 

By \cite{Clark1976}, an equilibrium of \eqref{EqT1C1}, $N^*$, is locally asymptotically stable if 
\begin{equation}\label{Clargen}
\left|p(N^*)+N^*p'(N^*)\right|+|m(\tau)|<1.
\end{equation}

A direct consequence is the following stability result. 
\begin{theorem}\label{thm:1}
    Consider \eqref{EqT1C1} for fixed $\tau\in \mathbb{N}$ and assume (H) given in \eqref{eq:H}. If 
    $p(0)<1-m(\tau)$, then $N^*_0=0$ is the only non-negative equilbrium and it is locally asymptotically stable. If $p(0)>1-m(\tau)$, then there exists a unique positive equilibrium $N^*_+$ and $N^*_0$ is unstable. 
\end{theorem}

In the case, where only the trivial equilibrium exists, we can strengthen this stability result.

\begin{theorem}\label{thm:2}
    Consider \eqref{EqT1C1} for fixed $\tau\in \mathbb{N}$ and assume (H) given in \eqref{eq:H}. If 
    $p(0)<1-m(\tau)$, then $N^*_0=0$ is the only non-negative equilibrium and it is globally asymptotically stable.
\end{theorem}

\begin{proof}
 By assumption, $b:=p(0)+m(\tau)<1$. Since, 
\begin{align*}
    N_{t+1}&=F(N_t,N_{t-\tau})=p(N_t)N_t+m(\tau)N_{t-\tau}\stackrel{(H)}{\leq}p(0)N_t+m(\tau)N_{t-\tau}\\
    &\leq \left(p(0)+m(\tau)\right)\max\{N_t, N_{t-\tau}\}=b\max\{N_t, N_{t-\tau}\}.
\end{align*}

Theorems \ref{thm:1} and \cite[Theorem~2]{LIZ2002} imply the global asymptotic stability of $N^*_0=0$.  
\end{proof}

We now consider the special case, with $m(\tau)$ given by \eqref{eq:mtau}. In that case, 
$$m'(\tau)= \overline{p}^{\tau+1} \left(b K e^{-K (\tau+1)} +\log(\overline{p}) \left(a -be^{-K(\tau+1)}\right) \right).$$
For sufficiently large $\tau$, $b K e^{-K (\tau+1)} +\log(\overline{p}) \left(a -be^{-K(\tau+1)}\right)\approx \log(\overline{p})a<0$. Thus, for sufficiently large $\tau$, $m(\tau)$ is decreasing in $\tau$. Then, as $\tau\to \infty$, $m(\tau)\to 0$.    

\begin{lemma}\label{lem:1}
    Consider \eqref{EqT1C1} and assume (H) given in \eqref{eq:H}.  Let $m(\tau)$ be given  by \eqref{eq:mtau}. Then, there exists a critical delay $\tau_c\geq 0$ such that $N^*_0=0$ is globally asymptotically stable for all $\tau\geq \tau_c$. 
\end{lemma}

\begin{proof}
By the above discussion on $m$, $m$ is eventually decreasing and $\lim_{\tau\to \infty}m(\tau)=0$. Thus, there exists $\tau_c$ such that $p(0)<1-m(\tau_c)$ and $p(0)<1-m(\tau)$ for all $\tau\geq \tau_c$. The claim then follows by Theorem~\ref{thm:2}
\end{proof}

\begin{theorem}\label{thm:NstarLAS1}
    Consider \eqref{EqT1C1} for fixed $\tau\in \mathbb{N}$ and assume (H) given in \eqref{eq:H} and  $p(0)>1-m(\tau)$. Then, there exists $\delta>0$ such that if $p(0)\in (1-m(\tau),1-m(\tau)+\delta)$, then the unique positive equilibrium $N^*_+$ is locally asymptotically stable. 
\end{theorem}

The proof is provided in Appendix~\ref{A:NstarLAS1}.

\subsection{Special Case: Beverton--Holt survival for mature individuals}

We now consider \eqref{EqT1C1} with $m(\tau)$ given in \eqref{eq:mtau} and assume that mature individuals follow the Beverton--Holt survival \eqref{eqn:BHmatsurvival}. That is, we consider 
\begin{equation}\label{EqT1BH}
N_{t+1}=F(N_t, N_{t-\tau})=\frac{1}{1+\alpha+\beta N_t}N_t + \left(a - b e^{-K(\tau+1)}\right)
\overline{p}^{\tau+1} N_{t-\tau}, \quad  N_i=N_i^0 \geq 0, \, i  \in \{-\tau, -\tau-1, \ldots, 0\}
\end{equation}
for $\alpha, \beta>0$. Because $p(N_t)$ given in \eqref{eqn:BHmatsurvival} satisfies (H) given in \eqref{eq:H} and the conditions for the application of Propositions~\ref{Proppos}, \ref{Propbound}, \ref{Prop:bound}, Lemma~\ref{lem:1}, and Theorems~\ref{thm:1}--\ref{thm:NstarLAS1} are satisfied, their corresponding results apply for this special case. Furthermore, $N^*_0=0$ is an equilibrium of \eqref{EqT1BH} and  a positive equilibrium $N^*_+$ satisfies
\begin{equation}\label{T1C1Nstar}
 N^*_+ =\frac{(1+\alpha) m(\tau)-\alpha}{\beta(1 -m(\tau))}.  
\end{equation}
 This, uniquely, determines the positive equilibrium that exists if and only if $(1+\alpha)m(\tau)>\alpha$, i.e. $p(0)>1-m(\tau)$.

Note that we can understand the positive equilibrium $N^*_+$ as a function of $\tau$, $N^*_+(\tau)$. Then, for given model parameters,  a positive equilibrium $N^*_+=N^*_+(\tau)$ only exists for those $\tau\in \mathbb{N}$, where $(1+\alpha)m(\tau)>\alpha$, i.e.,
\begin{equation}\label{condNstar}
\left(a - b e^{-K(\tau+1)}\right)
\overline{p}^{\tau+1}>\frac{\alpha}{1+\alpha}.
\end{equation}

Since the biological interpretation of the left-hand side is the fraction of offspring (per adult individual) that managed to survive to reach maturity, the above states that if the offspring (per adult individual) that reach maturity exceed the lower threshold of $\frac{\alpha}{1+\alpha}$, then there exists a unique positive equilibrium.

\begin{theorem}\label{Nstarex}
Consider \eqref{EqT1BH} and let $P:=\ln(1/\overline{p})>0$. 
If $$ \left(\frac{aP}{b(P+K)}\right)^{\frac{P}{K}}\frac{aK}{P+K}\leq \frac{\alpha}{1+\alpha},$$
then $N^*_+(\tau)$ is biologically irrelevant for any $\tau\in \mathbb{N}$. 
    Instead, if 
$$ \left(\frac{aP}{b(P+K)}\right)^{\frac{P}{K}}\frac{aK}{P+K}>\frac{\alpha}{1+\alpha}$$
then there exists $\tau_1<\tau_2$ such that $N^*_+(\tau)>0$ if and only if $\tau\in (\tau_1,\tau_2)$.  
Furthermore, if $m(0)<\frac{\alpha}{1+\alpha}$, then 
 $\tau_1>0$ and $N^*_+(\tau)>0$ if and only if $\tau\in (\tau_1,\tau_2)\subset \mathbb{R}^+$. Instead, if $m(0)\geq 0$, then $\tau_1\leq 0$ and $N^*_+(\tau)>0$ for $\tau\in (0,\tau_2)$. 
\end{theorem}

The proof of Theorem~\ref{Nstarex} is given in Appendix~\ref{Thm:Nstarex}. Theorem~\ref{Nstarex} suggests already that some delay may indeed be beneficial, especially when $m(0)<\frac{\alpha}{1+\alpha}$. Since, in that case, no positive equilibrium exists for $\tau\in [0,\tau_1]$ and only once the delay increases, a positive equilibrium population level exists.

Understanding the positive equilibrium, $N^*_+$, if it exists, as a function of $\tau$ also allows us to investigate the behavior (not only the existence) with respect to changing $\tau$. This is addressed in Theorem~\ref{ThmexNstar} with a proof provided in Appendix~\ref{proofThmexNstar}. 

\begin{theorem}\label{ThmexNstar}
Consider \eqref{EqT1BH} with model assumption (A) given in \eqref{eq:A}. Assume that the positive equilibrium $N^*_+(\tau)$ exists. Then, its value is increasing for $\tau<\widehat{\tau}$ and decreasing for $\tau>\widehat{\tau}$, where 
\begin{equation*}
\widehat{\tau}:=\frac{1}{K}\ln\left( \frac{bP+bK}{ae^{K}P}  \right)>0, \qquad \mbox{where} \, \, P:=\ln(1/\overline{p})>0.
\end{equation*}
\end{theorem}

Note that the case of $\tau=\widehat{\tau}$ is unlikely, since our recurrence in \eqref{EqT1BH} requires $\tau\in \mathbb{N}$ but $\widehat{\tau}\in \mathbb{R}_+$.

A graphical illustration of the behavior of $N^*_+(\tau)$ with respect to the maturation delay $\tau$, is depicted in Fig.~\ref{fig:Nstartau}. It clearly shows that there exist two critical values of delay. If the delay is too small, then the reproductive advantage is too small to secure the survival of the species. Instead, if the delay is too large, then too few individuals survive to maturity to secure the survival of the species. Only if the delay is between these two critical delay thresholds is the trade-off between reproductive strength and survival sufficiently balanced to secure the survival of the species. In fact, we see in Fig.~\ref{fig:Nstartau} that the positive equilibrium increases and then decreases, indicating an optimal delay for the trade-off problem. In case of an evolving species, one may predict that the species evolves towards the delay that results in the largest abundance level $N^*(\tau)$.

\begin{figure}[h!]
    \centering
    \includegraphics[width=0.5\linewidth]{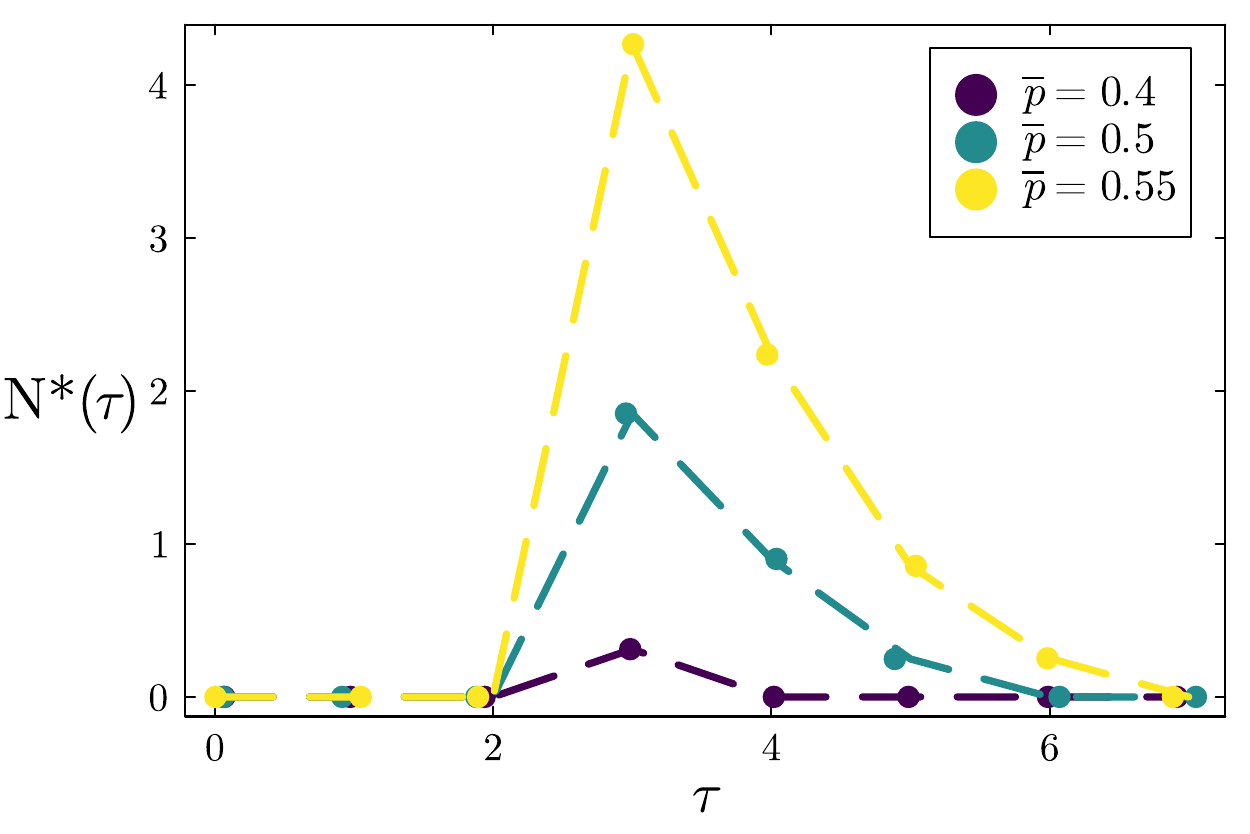}
    \caption{Dependence of $N^*$ on $\tau \in \mathbb{N}^+_0$ for different values of $\overline{p}\in (0,1)$. Parameter values are $K=1, b=200, a=10, \alpha=0.1, \beta=0.3$. $N^*(\tau)=\max\{N^*_0=0,N^*_+(\tau)\}$.}
    \label{fig:Nstartau}
\end{figure}

\begin{remark}
    The condition in Theorem~\ref{ThmexNstar} on $\tau$ could have been equivalently formulated with respect to the survival probability of immature individuals $\overline{p}$. That is, 
    \begin{equation*}
    \frac{\partial N_+^*}{\partial \tau}
\, \, \begin{cases} 
<0 & \overline{p}<\overline{p}_c(\tau)\\
=0 & \overline{p}=\overline{p}_c(\tau)\\
>0& \overline{p}>\overline{p}_c(\tau)
    \end{cases}\qquad \qquad \mbox{for}\qquad \overline{p}_c:=\exp\left\{\frac{-bK}{ae^{K(\tau+1)}-b}\right\}.
    \end{equation*}
Thus, the smaller the survival probability of immature individuals per unit-interval,  $\overline{p}$, the more likely it is that $\frac{\partial N_+^*}{\partial \tau}<0$, i.e., $N_+^*$ is decreasing as a function of $\tau$ with, by \eqref{T1C1Nstar}, $\lim_{\tau\to \infty}N_+^*(\tau)=\frac{-\alpha}{\beta}<0$ implying that it surpasses the trivial equilibrium, likely resulting in a transcritical bifurcation. 

\end{remark}

To study the local asymptotic stability, we recall that, by \cite{Clark1976}, an equilibrium $N^*\geq 0$ is  local asymptotically stability if 
\begin{equation}\label{LASClark}
1>\left|\frac{\partial F(N_t,N_{t-\tau})}{\partial N_t}\right| +\left|\frac{\partial F(N_t,N_{t-\tau})}{\partial N_{t-\tau}}\right|= Q(N^*)+m(\tau),
\end{equation}
where $m(\tau)$ is given by \eqref{eq:mtau} and
\begin{equation*}
Q(N^*)=\frac{1 + \alpha}{(1 + \alpha + \beta N^*)^2}.
\end{equation*}
This immediately implies that, since $Q(0)=\frac{1}{1+\alpha}$, that if $m(0)+\frac{1}{1+\alpha}<1$ (i.e., $m(0)<\frac{\alpha}{1+\alpha}$), then $N^*_0$ is locally asymptotically stable. In fact, we can strengthen this result.

\begin{theorem}\label{thm:N0GAS}
    Consider \eqref{EqT1BH}, for fixed $\tau\in \mathbb{N}_0$.   
 If $(1+\alpha)m(\tau)<\alpha$, then the only nonnegative equilibrium is $N^*_0$ and it is globally asymptotically stable for non-negative initial conditions. 
\end{theorem}

\begin{proof}
By assumption, $b:=Q(0)+m(\tau)=\frac
{1}{1+\alpha}+m(\tau)<\frac{1}{1+\alpha}+\frac{\alpha}{1+\alpha}=1$. Since, 
\begin{align*}
    N_{t+1}&=F(N_t,N_{t-\tau})=Q(N_t)N_t+m(\tau)N_{t-\tau}\stackrel{(H)}{\leq} Q(0)N_t+m(\tau)n_{t-\tau}\\
    &\leq \left(Q(0)+m(\tau)\right)\max\{N_t, N_{t-\tau}\}=b\max\{N_t, N_{t-\tau}\}.
\end{align*}

Equation \eqref{LASClark} and \cite[Theorem~2]{LIZ2002} imply the global asymptotic stability of $N^*_0=0$.  
    
\end{proof}

\begin{theorem}\label{thm:NstarLAS}
 Consider \eqref{EqT1BH}, for fixed $\tau\in \mathbb{N}_0$, with model assumption (A) given in \eqref{eq:A}. 
 If $(1+\alpha)m(\tau)>\alpha$, then $N^*_+>0$ given in \eqref{T1C1Nstar} exists and is locally asymptotically stable, while $N^*_0$ is unstable. 
\end{theorem}

The proof of this theorem is provided in Appendix~\ref{A:NstarLAS}. We can strengthen this result of local asymptotic stability and prove that, if the positive equilibrium exists, then it is indeed globally asymptotically stable as formulated in Theorem~\ref{thm:NstarGAS} (see proof  in Appendix~\ref{A:NstarGAS}). This proof is based on \cite[Theorem~1.15]{Ladas_2004} that we stated, for convenience of the reader, in Appendix~\ref{A:ThmLadas}. 

\begin{theorem}\label{thm:NstarGAS}
    Consider \eqref{EqT1BH}, for fixed $\tau\in \mathbb{N}_0$   and model assumption (A) given in \eqref{eq:A}.    
 If $(1+\alpha)m(\tau)>\alpha$, then there exists a unique positive equilibrium $N^*_+$ that is globally asymptotically stable. 
\end{theorem}

Theorems~\ref{thm:N0GAS} and \ref{thm:NstarGAS} complete the global dynamics for \eqref{EqT1BH}.



\begin{remark}
In light of Theorem~\ref{Nstarex} and Theorem~\ref{thm:NstarLAS}, we can conclude that there exist (at most) two transcritical bifurcations at $\tau_1$ and $\tau_2$ (as defined in Theorem~\ref{Nstarex}), determined by the solutions to $m(\tau)(1+\alpha)=\alpha$, where the  unique positive equilibrium $N^*_+$ and the extinction equilibrium $N^*_0$ interchange stability. If $m(0)<\frac{\alpha}{1+\alpha}$, then $\tau_1>0$ and there exist two transcritical bifurcations at $\tau_1$ and at $\tau_2$. As $\tau$ increases beyond $\tau_1$, the trivial equilibrium loses stability and the positive equilibrium $N^*_+$ becomes stable. At $\tau=\tau_2$, this is reversed and the positive equilibrium loses its stability to the trivial equilibrium. 
\end{remark}

We can also apply \eqref{LASClark} to study the stability of the extinction equilibrium for increasing $\tau$, as formulated in Theorem~\ref{thm:N0LAS} and proven in Appendix~\ref{A:thmN0LAS}. 
\begin{theorem}\label{thm:N0LAS}
Consider \eqref{EqT1BH} with model assumption (A) given in \eqref{eq:A}. Then, there exists $\tau_c>\frac{1}{K}\ln\left(\frac{b(P+K)}{aP}\right) -1$, where $P=\ln(1/\overline{p})$, such that for $\tau\geq \tau_c$, $N^*_0$ is locally asymptotically stable.
\end{theorem}

\subsection{Special Case: Ricker survival for mature individuals}

In this section, we consider that the survival of immature individuals is still constant. Therefore, $\overline{p}\in (0,1)$ for one time interval. However, now mature individuals have a Ricker survival function. We consider 
\eqref{EqT1C1} with $m(\tau)$ given in \eqref{eq:mtau} and assume that mature individuals follow  \eqref{eqn:Rmatsurvival} which yields the recursion
\begin{equation}\label{EqT1R}
N_{t+1}=e^{-\alpha-\beta N_t}N_t + \left(a - b e^{-K(\tau+1)}\right)
\overline{p}^{\tau+1} N_{t-\tau}, , \quad  N_i=N_i^0 \geq 0, \, i  \in \{-\tau, -\tau-1, \ldots, 0\}
\end{equation}
for $\alpha, \beta, a,b,K>0$ and $\overline{p}\in (0,1)$. As before, see \eqref{eq:growth}, we assume $a>be^{-K}$ and $m(\tau)\in (0,1)$. Again, the conditions for the application of Propositions~\ref{Proppos}, \ref{Propbound}, \ref{Prop:bound}, Lemma~\ref{lem:1}, and Theorems~\ref{thm:1}--\ref{thm:NstarLAS1} are satisfied so that these results apply for this special case.

Since $p(0)=e^{-\alpha}$, Theorem~\ref{thm:1} implies that for a positive equilibrium $N^*_+$ to exist, $1-e^{-\alpha}<m(\tau)$ must be satisfied. More precisely, if $1-e^{-\alpha}<m(\tau)$, then the positive equilibrium $N^*_+$ is given by  

\begin{equation}\label{eq:q}
N^*_+=\frac{-\ln(1-m(\tau))-\alpha}{\beta}>0.
\end{equation}

Thus, we can immediately formulate the following corollary as a direct consequence of Theorem~\ref{thm:1}. By Theorem~\ref{thm:2}, we can strengthen the stability result for the trivial equilibrium.

\begin{corollary}\label{cor:N0GASR}
Consider \eqref{EqT1R} for fixed $\tau\in \mathbb{N}$ with initial conditions in \eqref{eq:mtau}. 
If $1-e^{-\alpha}< m(\tau)$, then there exists a unique positive equilibrium $N^*_+$. Instead, if $1-e^{-\alpha}> m(\tau)$, then the only nonnegative equilibrium is the trivial equilibrium $N^*_0=0$, that is globally asymptotically stable.
\end{corollary}

By Lemma~\ref{lem:1}, the following result holds. 
\begin{corollary}
    Consider \eqref{EqT1R} for fixed $\tau\in \mathbb{N}$ and model assumption (A) given in \eqref{eq:A}. Then, there exists $\tau_c\geq 0$ such that $1-e^{-\alpha}\geq m(\tau)$ for all $\tau\geq \tau_c$. 
\end{corollary}

\begin{theorem}\label{thm:GASNstar}
     Consider \eqref{EqT1R} for fixed $\tau\in \mathbb{N}$ and model assumption (A) given in \eqref{eq:A}. If $1-e^{-\alpha}< m(\tau)<1-e^{-\alpha-2}$, then $N^*_+>0$ is locally asymptotically stable. Instead, if $m(\tau)>1-e^{-\alpha-2}$, then $N^*_+$ is unstable.
\end{theorem}

The proof is provided in Appendix~\ref{A:GASNstarR}. We conjecture that where $N^*_+$ is locally asymptotically stable, it is also globally asymptotically stable as suggested by Fig. \ref{fig:RC_stable} and Fig. \ref{fig:RC_unstable} in Appendix~\ref{A:GASNstarR_figs}.

\begin{remark}
Since $m(\tau)$, given in \eqref{eq:mtau}, is linear in the egg production potential ($a$) and proportional to the survival rate of immature individuals ($\overline{p}$), the conditions of Theorem~\ref{thm:GASNstar} 
can be equivalently formulated as 
$$\frac{1-e^{-\alpha}}{\overline{p}^{\tau+1}}+be^{-K(\tau+1)}<a<\frac{1-e^{-(\alpha+2)}}{\overline{p}^{\tau+1}}+be^{-K(\tau+1)}$$
or, in terms of $\overline{p}$, 
$$\frac{1-e^{-\alpha}}{a-be^{-K(\tau+1)}}<\overline{p}^{\tau+1}<\frac{1-e^{-(\alpha+2)}}{a-be^{-K(\tau+1)}}.$$
While the left inequalities guarantee the existence of a positive equilibrium, the right inequalities guarantee its local asymptotic stability. 
Note that the right inequality, in terms of $\overline{p}^{\tau+1}$, becomes more likely as $\tau$ increases since $\overline{p}\in (0,1)$. However, with increasing $\tau$, $\tau$ approaches the critical extinction threshold $\tau_c$, in which case $N^*_+$ is no longer positive. This clearly shows a sensitivity that albeit some delay is needed to guarantee the existence of the positive equilibrium, there might only be a small range, in which we can guarantee its local asymptotic stability. We hint on this sensitivity in Theorem~\ref{thm:PeriodD}, where we show the existence of a bifurcation that changes the stability of the positive equilibrium.  
\end{remark}

The condition $m(\tau)<1-e^{-\alpha}$ can also be expressed in terms of $\tau$, by noticing that 

\begin{equation*}
m'(\tau) \, \begin{cases} >0 & \tau>\tau_r\\
=0 & \tau=\tau_r\\
<0 & \tau>\tau_r,
\end{cases}\qquad \qquad \tau_r:=\frac{1}{K}\ln \left( \frac{b(P+K)}{aP}\right)-1, \qquad \qquad P:=\ln(\overline{p}^{-1})>0.
\end{equation*}

Thus, Corollary~\ref{cor:N0GASR} and Theorem~\ref{thm:GASNstar} can be rewritten as follows.

\begin{theorem}
    Consider \eqref{EqT1R} with model assumption (A) given in \eqref{eq:A}. Assume that $\tau_r>0$ and $\beta N^*_+<2$. 
    If $m(\tau_r)<1-e^{-\alpha}$, then $N^*_0$ is globally asymptotically stable for all $\tau\geq 0$. If $m(0)\leq 1-e^{-\alpha}< m(\tau_r)$, then 
        $N^*_0$ is globally asymptotically stable for $\tau\in [0,\tau_r)\cup (\tau_c,\infty)$. For $\tau\in (\tau_r, \tau_c)$, there exists a unique positive equilibrium $N^*_+$ that is locally asymptotically stable for $\tau\in (\tau_r,\tau_c)$. 
\end{theorem}

Theorem~\ref{thm:GASNstar} already suggests that $N^*_+$ may not be stable whenever it exists, in contrast to the case in Section 3.1 considering a Beverton--Holt survival for mature individuals. We take a closer look at the possible bifurcations in the next Theorems (see proof in Appendix~\ref{A:thmPeriodD} and Appendix~\ref{A:thmPeriodNS}).

\begin{theorem}\label{thm:PeriodD}
     Consider \eqref{EqT1R} with model assumption (A) given in \eqref{eq:A}. If $\tau$ is odd, then there exists a unique $a^*>0$ such that the system undergoes a period-doubling bifurcation at $a=a^*$. If $\tau$ is even, there is no period-doubling bifurcation for any values of $a>0$ such that $m(\tau)<1$. 
\end{theorem}

\begin{theorem}\label{thm:PeriodNS}
       Consider \eqref{EqT1R} with model assumption (A) given in \eqref{eq:A}. If $\tau$ is even, then there exists a 
       unique $a^*>0$ such that  the system undergoes a Neimark--Sacker bifurcation at $a=a^*$ for $N^*>0$. 
\end{theorem}

Theorems~\ref{thm:PeriodD} and \ref{thm:PeriodNS} state that for even $\tau$, the Ricker-Constant model exhibits a transcritical and a Neimark--Sacker bifurcation when changing the model parameter $a$ that describes the egg productivity. These bifurcation patterns are illustrated in Fig.~\ref{fig:RCtaueven_bifurcation}. 
However, when $\tau$ is odd, the Ricker-Constant model exhibits a transcritical bifurcation, one or two period doubling bifurcations, and a Neimark--Sacker bifurcation as the model parameter $a$ is being changed, see Figs.~\ref{fig:RCtau3_bifurcation} \& \ref{fig:RCtau5_bifurcation} for an illustration. Although the differences in bifurcation patterns, dependent on whether $\tau$ is even or odd, are hard to explain biologically, it is worth noticing that a period-doubling bifurcation and a Neimark--Sacker bifurcation may exhibit similarities in their appearances. A Neimark--Sacker bifurcation results in the existence of a  closed invariant curve but, since the time is discrete, the solution ``jumps'' along this closed curve and may do so in such large distances that it may first appear as a two-cycle. The latter is the result of a period-doubling bifurcation. 

Using XPPAUT, we analyze the effect of $\overline{p}$ and $a$ with a two parameter bifurcation plot. Because $\overline{p}$ is the survival of immature individuals and $a$ is a measure of the egg production potential, $\overline{p}$ and $a$ should, through different underlying mechanisms, affect the existence of a stable interior equilibrium and potentially whether a period-doubling or a Neimark--Sacker bifurcation occurs. Fig. \ref{fig:RC_apbifurcation}a shows the transcritical bifurcation and either first a period-doubling (odd $\tau$) or a Neimark-Sacker (even $\tau$) bifurcation for different values of $\tau$, $a$, and $\overline{p}$. Increasing either $\overline{p}$ or $a$, while keeping the other parameter constant, has the same general pattern of moving the dynamics through first a transcritical bifurcation and then inducing oscillations. Increasing $\tau$ shifts these bifurcation curves up and to the left, with the effect of increasing the minimum $\overline{p}$ and decreasing the minimum $a$ that is needed for the existence of an interior equilibrium (or inducing oscillations). Because of this interaction, depending on the parameters, increasing $\tau$ can change the equilibrium from zero to positive (A), or induce oscillations where no oscillations were occurring (B), or remove oscillations when oscillations were occurring (C). Finally, for sufficiently large $\tau$ values, continuing to increase $\tau$ changes the equilibrium from positive to zero (D), causing the species to go extinct as already established in Lemma~\ref{lem:1} (Fig.~\ref{fig:RC_apbifurcation}a). We show this effect with two orbit diagrams, where $a=13.0$ \& $\overline{p}=0.57205$ and $a=13.0$ \& $\overline{p}=0.35$ (Fig.~\ref{fig:RC_apbifurcation}b\&c).

\begin{figure}[h!]
    \centering
        \includegraphics[width=\textwidth]{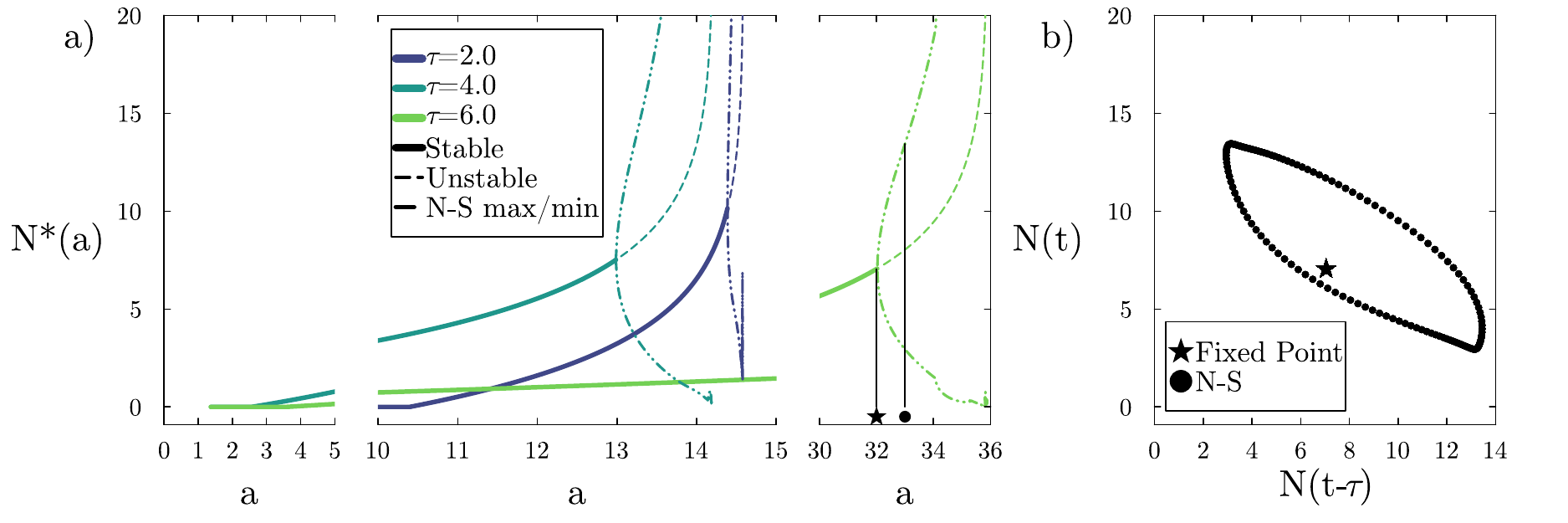}
    \caption{a) Bifurcation diagram for the Ricker-Constant model with $\tau=2,4,6$. Vertical lines with a star and circle denote the $a$ values used in the time embedding plot in b). b) Time embedding for the Ricker-Constant model with $\tau=6$ and $a=32.0$ (Star) or $a=33.0$ (Circle). In both plots, $\alpha=0.1$, $\beta=0.3$, $K=1.0$, and $b=200.0$}
    \label{fig:RCtaueven_bifurcation}
\end{figure}

\begin{figure}[h!]
    \centering
        \includegraphics[width=\textwidth]{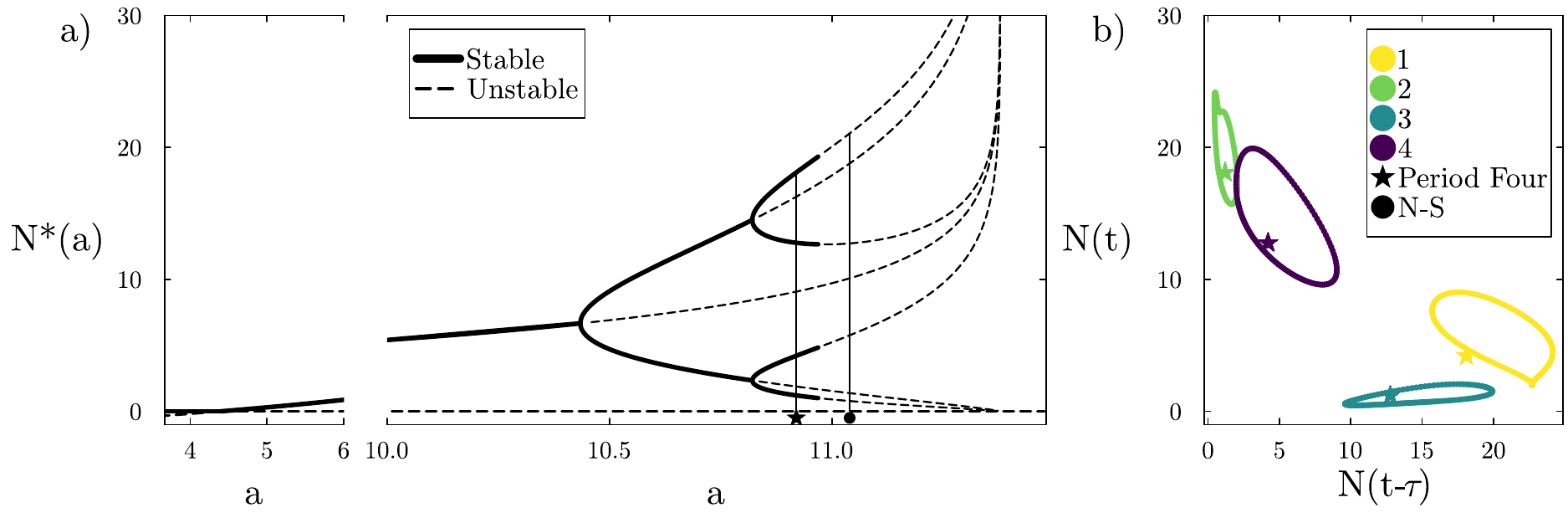}
    \caption{a) Bifurcation diagram for the Ricker-Constant model with $\tau=3$. Vertical lines with star and circle denote the $a$ values used in the time embedding plot in b). b) Time embedding for the Ricker-Constant model with $\tau=3$ and $a=10.92$ (Star) or $a=11.04$ (Circle). Colour (yellow, green, blue, purple) refers to the ordering of $N(t)$ in the sequence. The time embedding plots were from the last 1000 time units from a 1,000,000 long simulation. In both plots, $\alpha=0.1$, $\beta=0.3$, $K=1.0$, and $b=200.0$}
    \label{fig:RCtau3_bifurcation}
\end{figure}

\begin{figure}[h!]
    \centering
            \includegraphics[width=\textwidth]{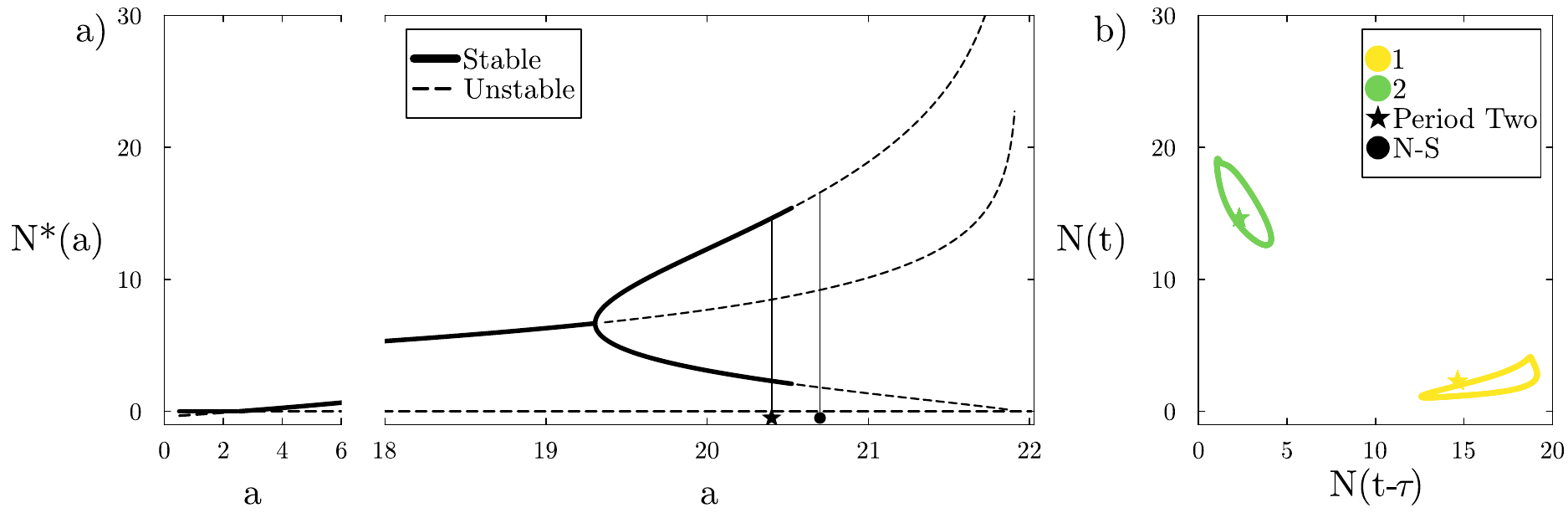}
    \caption{a) Bifurcation diagram for the Ricker-Constant model with $\tau=5$. Vertical lines with star and circle denote the $a$ values for the time embedding plot in b). b) Time embedding for the Ricker-Constant model with $\tau=5$ and $a=20.2$ (Star) or $a=20.7$ (Circle). Colour (yellow, green) refers to the ordering of $N(t)$ in the sequence. The time embedding plots were from the last 1000 time units from a 1,000,000 long simulation. In both plots, $\alpha=0.1$, $\beta=0.3$, $K=1.0$, $b=200.0$}
    \label{fig:RCtau5_bifurcation}
\end{figure}

\newpage

\begin{figure}[h!]
    \centering
            \includegraphics[width=\textwidth]{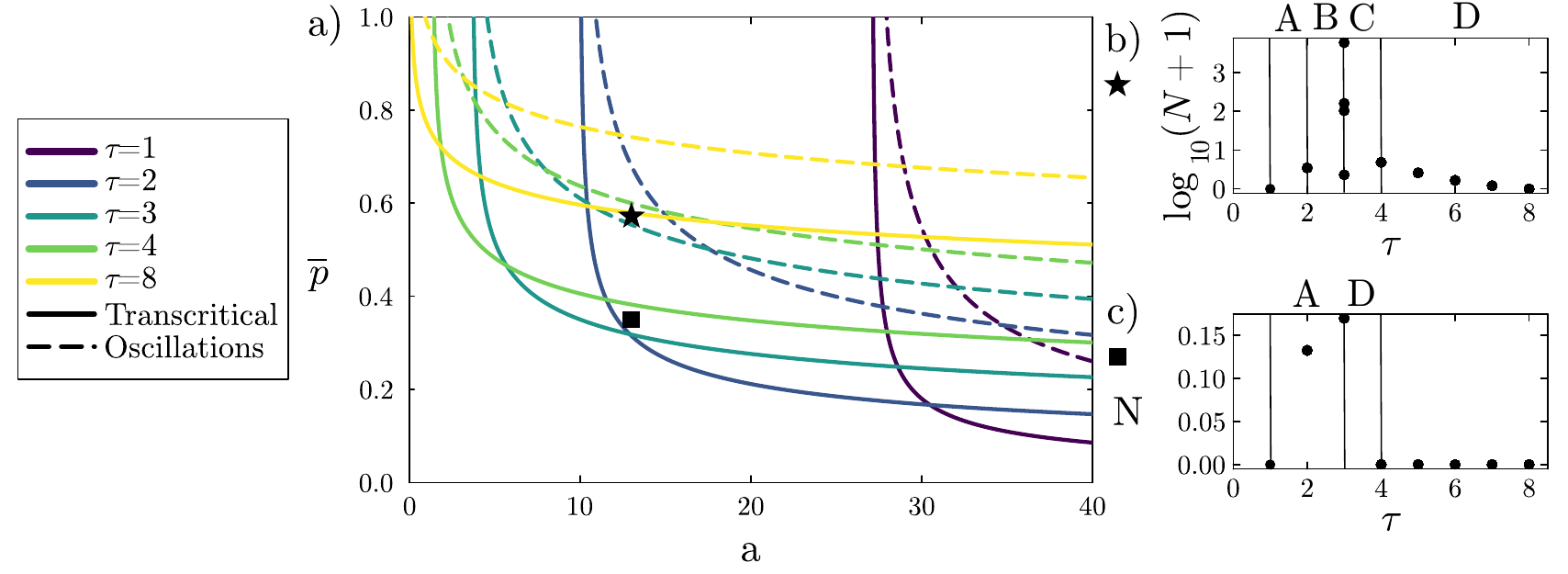}
    \caption{a) Two parameter ($a$ \& $\overline{p}$) bifurcation diagram for the Ricker-Constant model. Note that the oscillations are either the Neimark--Sacker bifurcation for even $\tau$ or the first period doubling bifurcation for odd $\tau$. The star refers to $a=13.0$ \& $\overline{p}=0.57205$. The square refers to $a=13.0$ \& $\overline{p}=0.35$. b) \& c) Orbit diagrams with changing $\tau$ corresponding to the star and square parameter combinations respectively. In b) (star), $\overline{p}=0.57205$, and in c) (square), $\overline{p}=0.35$. Orbit diagrams were created from the last 50 time units of a 1,000,000 long simulation. For all plots, $\alpha=0.1$, $\beta=0.3$, $K=1.0$, $b=200.0$. A, B, C, D refer to the different scenarios when increasing $\tau$: (A) equilibrium changes from zero to positive, (B) dynamics become oscillatory, (C)  dynamics returns equilibrium, and (D) decreasing to extinction.}
    \label{fig:RC_apbifurcation}
\end{figure}

\newpage

\section{Scenario ii) \textemdash Cohort-density-dependent survival of immature individuals }

Thus far, we considered the special case where the survival of immature individuals, over one time period $(t,t+1)$ is constant, determined by $\overline{p}\in (0,1)$. In that case, the survival of the entire immature period from $t-\tau$ to $t+1$ is $\overline{p}^{\tau+1}$. This assumed that any density-dependence is negligible. We now extend this approach by considering a within cohort density-dependence effect. Here, we assume that individuals are competing over resources only within their cohort and thus do not compete or interact with individuals in other cohorts. This is a reasonable first approximation of ontogenetic niche differentiation. Although inter-cohort density-dependence regularly occurs in species (e.g.,  \cite{gamelonDensityDependenceAgestructured2016}), within-cohort density-dependence is particularly important. For example, within-cohort density-dependence was found to impact the population patterns of Atlantic cod (\textit{Gadus morhua} L.) \cite{stensethDynamicsCoastalCod1999}. We leave future work to explore more complicated within- and inter-cohort density-dependence effects and consider henceforth that within cohort density-dependence effects are much stronger than inter-cohort density-dependence effects.

As in the previous section, we again consider a Beverton--Holt and a Ricker survival for the immature individuals. We remind the reader that, in this case, we do consider the same survival for the mature and immature individuals. That is, if the immature individuals follow a Beverton--Holt (Ricker) survival function over one time interval $(t,t+1)$, then  so do the mature individuals.

\subsection{Special Case: Beverton--Holt survival for immature individuals}

We consider the general delay population model \eqref{EqmainT1}, where the survival of mature individuals follows a Beverton--Holt survival, as given in \eqref{eqn:BHmatsurvival}. As discussed in \cite{StWodelay2}, if the immature individuals also have a Beverton--Holt survival over one time interval, then  the survival fraction of newly born individuals at time $t-\tau$,  $g(\tau)N_{t-\tau}$, is derived in \cite{StWodelay2} as \eqref{eq:BHimmsurvival}. 
Then, \eqref{EqmainT1} is of the form  
 \begin{equation}\label{EqT2BH}
N_{t+1}=F_2(N_t,N_{t-\tau})=\frac{1}{1+\alpha+\beta N_t}N_t + 
m(\tau, N_{t-\tau})N_{t-\tau}, , \quad  N_i=N_i^0 \geq 0, \, i  \in \{-\tau, -\tau-1, \ldots, 0\}
\end{equation}
with $m(\tau, N_{t-\tau})$ given in \eqref{eq:mtaugen}, that is, 
\begin{equation}\label{eq:mtaudenBH}
m(\tau, N_{t-\tau})= g(\tau)\widetilde{p}(\tau,N_{t-\tau})=   \frac{D}{D(1+D)^{\tau+1}+((1+D)^{\tau+1}-1)Cg(\tau)N_{t-\tau}} g(\tau),
\end{equation}
for $g(\tau)$ given in \eqref{eq:growth}, assuming $a>be^{-K}$. As before, we consider \eqref{EqT2BH} with 
\begin{equation}\label{IniC2}
    N_{-r}\geq 0, \quad \mbox{for} \quad r\in \{0, \ldots, \tau\}, \quad \max\{N_{0}, \ldots, N_{-\tau}\}>0. 
\end{equation}
As before, if all initial conditions are zero, that is, $N_{-r}=0$ for all $r\in \{0, \ldots, \tau\}$, then $N_t=0$ for all $t\geq 0$. Furthermore, the extinction $N^*_0=0$ is still an equilibrium. Since a positive equilibrium $N^*_+$ must satisfy the equation, for $x=N^*_+$, 
$$ 1-\frac{1}{1+\alpha+\beta x}= 1-p(x)=m(\tau, x)=
\frac{D}{D(1+D)^{\tau+1}+((1+D)^{\tau+1}-1)Cg(\tau)x} g(\tau),$$  and $1-p(x)$ is increasing in $x$ and $m(\tau, x)$ is decreasing in $x$ for fixed $\tau\in \mathbb{N}_0$, there can be at most one intersection. We therefore have immediately the following lemma regarding the existence of non-negative equilibria. 

\begin{lemma}\label{lem:BHBHinequlity}
Consider \eqref{EqT2BH} with \eqref{eq:mtaudenBH} and assume \eqref{IniC2}. If 
\begin{equation}\label{eq:condN0only}
    1-\frac{1}{1+\alpha}\geq \frac{g(\tau)}{(1+D)^{\tau+1}},
\end{equation}
where $g(\tau)$ is given in \eqref{eq:growth}, then 
the only non-negative equilibrium is the trivial equilibrium $N^*_0=0$. If, instead, 
\begin{equation}\label{eq:condNstar}
    1-\frac{1}{1+\alpha}< \frac{g(\tau)}{(1+D)^{\tau+1}},
\end{equation}
then there exists a unique positive equilibrium $N^*_+$.
\end{lemma}

A positive equilibrium, $N^*_+$, satisfies 
$$\frac{\alpha+\beta N^*_+}{1+\alpha+\beta N^*_+}=\frac{Dg(\tau)}{D(1+D)^{\tau+1}+((1+D)^{\tau+1}-1)Cg(\tau)N^*_+},$$
which after rearranging is equivalent to $a_0+a_1N^*_++a_2(N^*_+)^2=0$, where 
\begin{align*}
    a_0:&=D (\alpha ((1+D)^{\tau+1}- g(\tau)) - g(\tau))\\
    a_1:&=-((\alpha C + \beta D) g(\tau)) + (1+D)^{\tau+1} (\beta D + \alpha C g(\tau))\\
    a_2:&=((1+D)^{\tau+1}-1)\beta C g(\tau)>0.
\end{align*}
Thus, if existent, the unique positive equilibrium $N^*_+$ is given by 
\begin{equation}\label{eq:NstarDBH}
    N^*_+:=\frac{-a_1+\sqrt{a_1^2-4a_0a_2}}{2a_2}.
\end{equation}

We conjecture that when a positive equilibrium exists, this equilibrium also exhibits a bell-shape pattern similar to the BH-Constant model (Fig.~\ref{fig:Nstartau}) as indicated by Fig.~\ref{fig:BHBH}. A graphical illustration of the behavior of $N^*_+(\tau)$ with respect to the maturation delay $\tau$, immature density-dependence ($C$), mature density-dependence ($\beta$), immature mortality ($D$), and mature mortality ($\alpha$) is depicted in Fig.~\ref{fig:BHBH}. Similar to the BH-Constant model, Fig.~\ref{fig:BHBH} shows that there exist two critical values of delay. Only if the delay is between the two critical delay thresholds is the trade-off between reproductive strength and survival sufficiently balanced to secure the survival of the species. Again, we see that the positive equilibrium increases and then decreases, indicating an optimal delay for the trade-off problem. In terms of other parameters affecting the existence and magnitude of the positive interior equilibrium, we examine $C$, $\beta$, $D$, and $\alpha$ because there may be differential impacts of density-dependence and mortality for the immature individuals versus the mature individuals. As per Lemma \ref{lem:BHBHinequlity} we see that only $D$ and $\alpha$ affects what $\tau$ value allows the existence of the positive interior equilibrium. Decreasing $C$, $\beta$, $D$, or $\alpha$ lead to higher equilibrium N values. Decreasing $C$ and $D$ has a larger effect on the equilibrium compared to decreasing $\alpha$ and $\beta$ (Fig. \ref{fig:BHBH}).

\begin{figure}[H]
    \centering
    \includegraphics[width=0.8\linewidth]{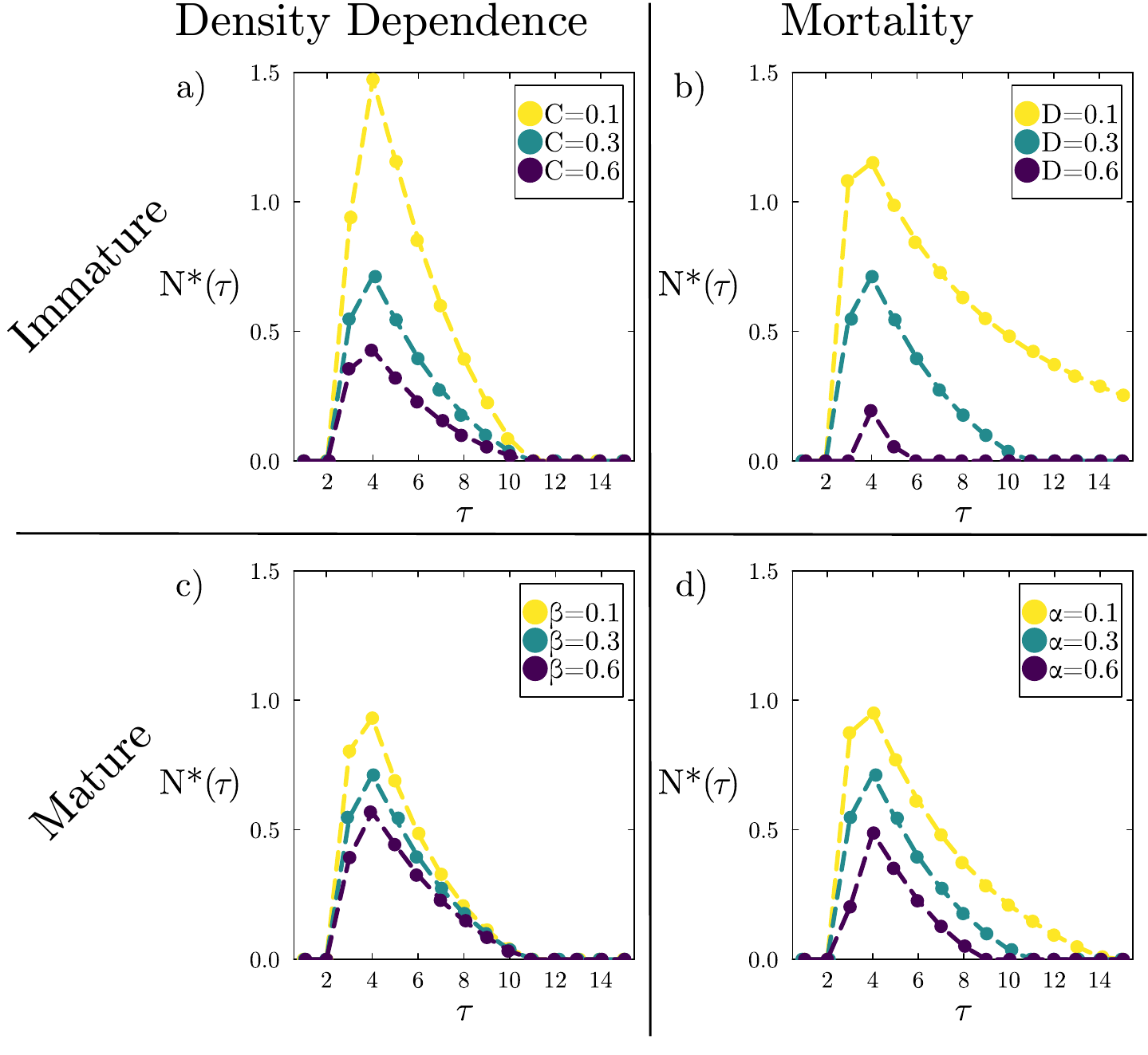}
    \caption{Dependence of $N^*$ on $\tau$ for different values of C (subplot a), D (subplot b), $\alpha$ (subplot c), and $\beta$ (subplot d). For all plots, base parameters are $C=0.3$, $D=0.3$, $\alpha=0.3$, $\beta=0.3$, $a=5.0$ $b=200.0$, and $K=1.0$}
    \label{fig:BHBH}
\end{figure}

We continue our investigation of \eqref{EqT2BH} by analyzing its local and global dynamics where Lemmas and Theorems~\ref{T2N0}-\ref{thm:GASNstarBH} complete the global analysis of the BH--BH model for fixed  $\tau\in \mathbb{N}$. Based on \cite{Clark1976}, an equilibrium $N^*$ of \eqref{EqT2BH} is locally asymptotically stable if 
$$\left|\frac{\partial F_2(u,v)}{\partial u}\right|_{u=v=N^*}+\left|\frac{\partial F_2(u,v)}{\partial v}\right |_{u=v=N^*}<1,$$
which implies, in our case that 
\begin{equation}\label{T2Clark}
    \frac{1+\alpha}{(1+\alpha+\beta N^*)^2} + 
\frac{D^2 g(\tau)(1+D)^{\tau+1} }{(D(1+D)^{\tau+1}+((1+D)^{\tau+1}-1)Cg(\tau)N_{t-\tau})^2} <1.
\end{equation}

This immediately allows the formulation of the following result. 
\begin{lemma}\label{T2N0}
    Consider \eqref{EqT2BH} with \eqref{eq:mtaudenBH} and assume \eqref{IniC2}. If strict inequality in \eqref{eq:condN0only} holds, then $N^*_0$ is locally asymptotically stable. In the case of equality, $N^*_0$ is stable.   
\end{lemma}

\begin{proof}
    If the inequality in \eqref{eq:condN0only} is strict, the result follows immediately. If instead equality holds, then, $g(\tau)=\frac{\alpha (1+D)^{\tau+1}}{1+\alpha}$, in which case the recurrence \eqref{EqT2BH} is of the form 
    $$N_{t+1}=\frac{N_t}{1+\alpha+\beta N_t}+\frac{D\alpha N_{t-\tau}}{D(1+\alpha)+((1+D)^{\tau+1}-1)C\alpha N_{t-\tau}}.$$
Let $N_t$ be a solution to the initial conditions with $N_i\leq \delta$ for $i\in \{-\tau, \ldots, 0\}$. 
Note that  $N_{t+1}=F_2(N_t,N_{t-\tau})=F_{21}(N_t)+F_{22}(N_{t-\tau})$ for $F_{21}(N_t)$ and $F_{22}(N_{t-\tau})$ both increasing  in their respective variables. 
To show that the solution converges to $N^*_0=0$, we let $\overline{N}_t=\max\{N_t, N_{t-1},\ldots, N_{t-\tau}\}$. Then, 
\begin{align*}
N_{t+1}&\leq \frac{\overline{N}_t}{1+\alpha+\beta \overline{N}_t}+\frac{D\alpha \overline{N}_{t}}{D(1+\alpha)+((1+D)^{\tau+1}-1)C\alpha \overline{N}_{t}}\\
&\leq  \frac{\overline{N}_t}{1+\alpha}+\frac{D\alpha \overline{N}_{t}}{D(1+\alpha)}=\frac{\overline{N}_t}{1+\alpha}+\frac{\alpha \overline{N}_{t}}{(1+\alpha)}=\overline{N}_t,
\end{align*}
where the second inequality is only an equality if $\overline{N}_t=0$. Thus, $N^*_0=0$ is stable. 
\end{proof}

Understanding $N_{t+1}=F_2(N_t,N_{t-\tau})=F_{21}(N_t)+F_{22}(N_{t-\tau})$ and noticing that $F_{21}(N_t)$ and $F_{22}(N_{t-\tau})$ are both increasing functions in their respective variables, we utilize \cite[Theorem~1.15]{Ladas_2004} to prove the following result. 

\begin{theorem}
    Consider \eqref{EqT2BH} with \eqref{eq:mtaudenBH} and assume \eqref{IniC2}. If \eqref{eq:condN0only} holds, then $N^*_0$ is globally asymptotically stable. 
\end{theorem}

\begin{proof}
    Since $F_{2}$ is increasing in both of its variables, the set $[0, \max_{r\in \{\tau, \ldots, 0\}}\{N_{-r}\}]$ is positively invariant,  and the only solutions to 
    $$r=F_2(r,r)\qquad \mbox{and}\qquad R=F_2(R,R),$$
    is $r=R=0$, solutions to non-negative initial conditions converge to $N^*_0=0$. Since solutions are stable by Lemma~\ref{T2N0}, $N^*_0$ is globally asymptotically stable for non-negative initial conditions.     
\end{proof}

\begin{lemma}
    Consider \eqref{EqT2BH} with \eqref{eq:mtaudenBH} and assume \eqref{IniC2}. If \eqref{eq:condNstar} holds, then the unique positive equilibrium $N^*_+$, given in \eqref{eq:NstarDBH}, is locally asymptotically stable. 
\end{lemma}

\begin{proof}
Since \eqref{eq:condNstar} holds, there exists a unique positive equilibrium $N^*_+$ that satisfies 
$$1-\frac{1}{1+\alpha+\beta N^*_+}=\frac{Dg(\tau)N^*_+}{D(1+D)^{\tau+1}+((1+D)^{\tau+1}-1)Cg(\tau)N^*_+}.$$
For $N^*_+>0$, we have 
\begin{align*}
&\left|\frac{\partial}{\partial N_t}F_2(N_t,N_{t-\tau})\right|+\left|\frac{\partial}{\partial N_{t-\tau}}F_2(N_t,N_{t-\tau})\right|=\frac{1+\alpha}{(1+\alpha +\beta N^*_+)^2}+\frac{D^2g(\tau)(1+D)^{\tau+1}}{(D(1+D)^{\tau+1}+((1+D)^{\tau+1}-1)Cg(\tau)N_+^*)^2}\\
&=\frac{1+\alpha}{(1+\alpha +\beta N^*_+)^2}+\frac{D(1+D)^{\tau+1}}{(D(1+D)^{\tau+1}+((1+D)^{\tau+1}-1)Cg(\tau)N_+^*)}\left(1-\frac{1}{1+\alpha+\beta N^*_+}\right)\\
&< \frac{1+\alpha}{(1+\alpha +\beta N^*_+)^2}+\left(1-\frac{1}{1+\alpha+\beta N^*_+}\right)=1-\frac{\beta N^*_+}{(1+\alpha+\beta N^*_+)^2}<1,
\end{align*}
satisfying \eqref{T2Clark}, implying the local asymptotic stability of $N^*_+$, whenever it exists. 
\end{proof}

Again, utilizing the monotonicity property of $F_{2}$ in both of its variables, we can employ again \cite[Theorem~1.15]{Ladas_2004} to complete the global dynamics (see Appendix~\ref{A:thmGASNstarBH}).

\begin{theorem}\label{thm:GASNstarBH}
      Consider \eqref{EqT2BH} with \eqref{eq:mtaudenBH} and assume \eqref{IniC2}. If \eqref{eq:condNstar} holds, then the unique positive equilibrium $N^*_+$ is globally asymptotically stable.
\end{theorem}

\subsection{Special Case: Ricker survival of immature individuals }

We now assume that the density-dependent survival of immature individuals follows a Ricker survival so that an immature cohort $w_t$ has the survival of the form \eqref{eqn:Rmatsurvival} for one time period $(t,t+1)$. To obtain the entire survival fraction of immature individuals from birth at time $t-\tau$ to maturity at time $t+1$, the system \eqref{eqn:RickerRicker} is used. Apart from Theorem \ref{thm:RR_PD} proving there are no period doubling bifurcations for even $\tau$ (see Appendix \ref{A:thmRR_PD}), we rely on computational simulations because analytical results require an analytical expression of $\widetilde{p}(\tau,N_{t-\tau})$ that relies on an explicit solution for the Ricker recurrence. Unlike for the Beverton--Holt model, we do not have an explicit solution for the Ricker recurrence. 

From scenario i), particularly the Ricker-Constant case, we know that the parameter $a$ (the measurement of egg production potential) has an important effect on the dynamics of our model (see Figs.~\ref{fig:RCtaueven_bifurcation}--\ref{fig:RCtau5_bifurcation}). We also know from the BH-BH model, that there are some differential impacts of the immature ($C$) and mature ($\beta$) density-dependence as illustrated in Fig.~\ref{fig:BHBH}. This motivated our two parameter bifurcation plots (created with XPPAUT) of $a$ \& $C$ and $a$ \& $\beta$ (Fig.~\ref{fig:RR_aCbeta}). 

Fig.~\ref{fig:RR_aCbeta}a) \& b) shows that for our parameter values, when manipulating $a$ (egg production potential) with either $C$ (immature density-dependence) or $\beta$ (mature density-dependence), the Ricker-Ricker model exhibits a transcritical and a Neimark--Sacker bifurcation. For the transcritical line, from left to right, the stable equilibrium goes from $0$ to $N^*_+$. For the Neimark--Sacker curve, from left to right, a stable positive interior equilibrium becomes unstable. For both $C$ and $\beta$, increasing $\tau$ reduces the value of $a$ at which the transcritical bifurcation occurs indicated by the movement of the vertical solid line (threshold line for the transcritical bifurcation) to the left. For a specific $\tau$ value, the value of $a$ at which the bifurcation occurs (solid line) is constant for all values of $C$ and $\beta$. This implies that the critical value at which the transcritical bifurcation occurs is not dependent on the density-dependent competition factors $C$ and $\beta$ but is dependent on the egg production potential parameter $a$. In other words, with a longer maturation delay, lower egg production potential is required for a positive population because the longer maturation delay increases the total egg production.

In contrast, the value at which a Neimark--Sacker bifurcation (see dashed curve) occurs is a function of the parameters $\tau$ (maturation delay), $a$ (egg production potential), $C$ (immature density-dependence) and $\beta$ (mature density-dependence). Generally, but with some nuances for both $C$ and $\beta$, increasing $\tau$ increases the range of $C$ (or $\beta$) and decreases the minimum $a$ where oscillations occur, with further increases in $\tau$ having the opposite effect. Because of this interaction, depending on the parameter values, increasing $\tau$ can increase $N^*_+$ away from zero, then induce oscillations, then again stop the oscillations, and finally reduce $N^*_+$ to zero (see Fig.~\ref{fig:RR_aCbeta}c--g). Biologically, increasing the maturation delay through the balance of total egg production and total immature losses can increase the population size, make the population dynamics oscillatory or back to stable, with extinction occurring at a long enough maturation delay.  Similarly, immature and mature density-dependence can also increase population sizes and affect population stability.

For a deeper comparison of the differential impacts of immature and mature density-dependence, we investigate the strength of immature density-dependence ($C$) relative to the mature density-dependence ($\beta$) by plotting orbit diagrams for different relations of $C$ and $\beta$. More precisely, we consider $C=\beta$,  $C>\beta$, and $C<\beta$ (see Fig.~\ref{fig:RR_aCbeta}c-g). Generally, maintaining $\tau$ and $a$, while weakening $C$ relative to $\beta$, reduces the oscillations highlighted by the reduced span of $\tau$ values that exhibit oscillations. Strengthening $C$ relative to $\beta$ has the opposite effect causing an increase in oscillations. Note, the effect of reducing oscillations by making immature density-dependence weaker relative to mature density-dependence is not due to reducing $C$ because keeping $C=\beta$ while decreasing both does not remove the oscillations (Appendix \ref{A:cequalbeta} Fig. \ref{fig:RR_Cequalbeta}). Biologically, when immature density-dependence is stronger than mature density-dependence the amplitude of population oscillations are larger compared to when mature density-dependence is stronger than immature density-dependence.

\begin{figure}[h!]
    \centering
    \includegraphics[width=\textwidth]{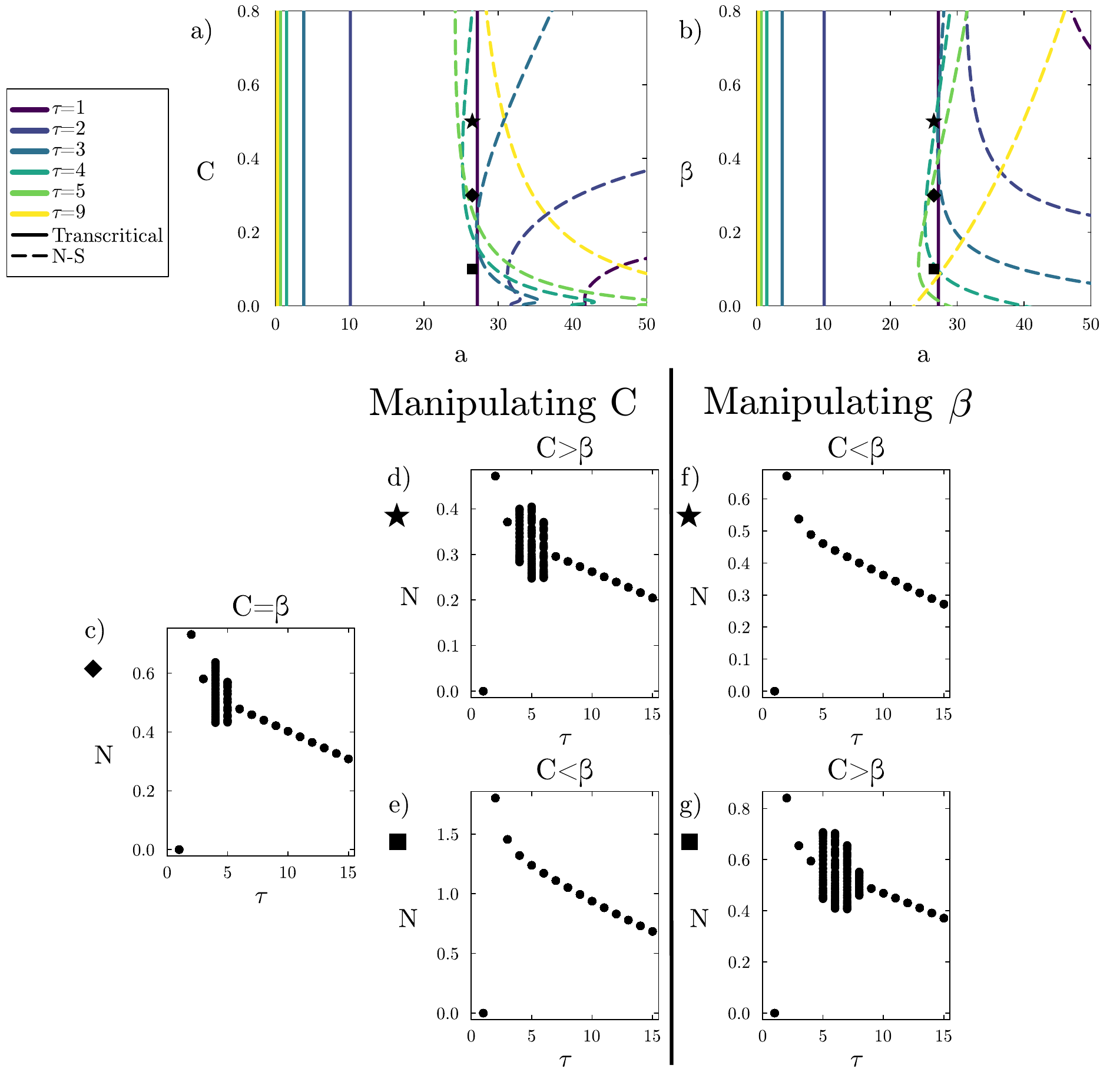}
    \caption{
    a) Two parameter bifurcation ($a$ and $C$) for the Ricker-Ricker model and different $\tau$ values. $C$ is the within immature cohorts density-dependence parameter. For the transcritical line, from left to right stable equilibrium goes from $0$ to $N^*>0$. For the Neimark-Sacker curve, from left to right a stable positive interior equilibrium becomes unstable. The star, diamond, and square refer to $a,C=26.5,0.6$, $a,C=26.5,0.3$, and  $a,C=26.5,0.1$ respectively (used in our c), d), \& e)). b) Two parameter bifurcation ($a$ and $\beta$) for the Ricker-Ricker model and different $\tau$ values. $\beta$ is the within mature individuals density-dependence parameter. The star, diamond, and square refer to $a,\beta=26.5,0.6$, $a,\beta=26.5,0.3$, and $a,\beta=26.5,0.1$ respectively (used in c), f), g). c) Ricker-Ricker orbit diagram for $a,C,\beta=26.5,0.3,0.3$. d) Ricker-Ricker orbit diagram for $a,C=26.5,0.6$. e) Ricker-Ricker orbit diagram for $a,C=26.5,0.1$. f) Ricker-Ricker orbit diagram for $a,\beta=26.5,0.6$. g) Ricker-Ricker orbit diagram for $a,\beta=26.5,0.1$. All orbit diagrams were the last 50 time units of 1,000,000 time units long simulations. For all plots, the base parameters are $C=0.3$, $\beta=0.3$,  $D=0.1$, $\alpha=0.1$, $a=26.5$, $b=200.0$, $K=1.0$.}
    \label{fig:RR_aCbeta}
\end{figure}

The diagrams in Fig.~\ref{fig:RR_aCbeta} indicate the complexity of the Ricker-Ricker delay recurrence model, particularly in comparison to the BH-BH case discussed in Section 4.1. This is not surprising, as the single species Ricker model already exhibits complex dynamics. Hence, it is reasonable to expect similarly complex dynamics when considering a repeated composition of Ricker functions to determine the survival of immature individuals. Potentially surprisingly however is that our simulations did not reveal chaotic behavior but rather a Neimark--Sacker bifurcation, before returning to the stability of the equilibrium $N^*_+$. Nevertheless, only a thorough parameter sweep would be able to reveal such chaotic behavior.  

\section{Conclusion}
Through refining the delay discrete population model in \cite{StWodelay2} to include  a maturation 
delay-dependent fecundity, we tested the effect of a trade-off caused by a longer maturation delay that, on the one hand increased fecundity but, on the other hand, increased the total losses of immature individuals. The benefit of higher egg production in the case of a longer delay was motivated by the fact that length and egg production are proportional, where the von Bertalanffy curve was utilized to describe the relation between age and length.  In addition to general results for such discrete delay population models with trade-off, such as the existence of a critical delay threshold that predicts the extinction of the species, 
we specifically focused on two classical survival functions \textemdash the Beverton--Holt and the Ricker survival. We distinguished between the mature survival function and the immature survival function so that the two classical survival forms resulted in 4 different models: i) the BH-Constant model, where the mature individuals follow a Beverton--Holt survival function and the immature individuals have a constant survival chance for one time period, ii) the Ricker-Constant model, where the mature individuals have now a  survival probability determined by the Ricker equation. The density-dependent survival cases for immature individuals led to the last two cases: iii) the BH-BH model, where the survival of the immature and mature individuals over one time period $(t,t+1)$ is described by a Beverton--Holt function, and iv) the Ricker-Ricker model, where the survival of the immature and mature individuals over one time period $(t,t+1)$ is described by a Ricker function. 

Across all four models, we observed an optimal maturation delay that was positive, where the population size at equilibrium was maximized. This pattern differs to the results obtained in the previous discrete delay population models \cite{StWodelay2, StWodelay3} that show the equilibrium population value decreasing with the delay. However, consistent with the previous predictions in \cite{StWodelay2, StWodelay3} that did not consider a trade-off, a further increase in the maturation delay leads to extinction. Furthermore, our finding of an optimal maturation delay reflects similar findings in the ecological-evolutionary literature. In our work, we considered the benefit of a longer delay to be higher fecundity, that is also found in nature. For example, \textit{Drosophila} exhibit lower fecundity with a shortened maturation delay  \cite{chippindaleExperimentalEvolutionAccelerated1997}. In \cite{bellCostsReproductionTheir1980}, the authors argue that delayed maturation would be selected for in nature, if fecundity increases with age. In that case, the optimal maturation age is achieved, where present reproduction causes a decrease in potential future fecundity. Other benefits to maturation delays that we did not consider in this work, but may be of interest in future work, are larger clutch size \cite{fordExperimentalStudyTradeOffs1994a}, higher mature or immature survival \cite{zeraPhysiologyLifeHistory2001a}, and escape from predation or harvesting of mature individuals \cite{poosHarvestinducedMaturationEvolution2011}. Generally across these modeling or empirical studies, we see an optimal maturation delay, where at first the benefits of the maturation delay outweigh the negative consequences until a maturation delay where the negative consequences start to outweigh the benefits.

In the models that consider the Ricker survival form in the mature population, with either constant or Ricker density-dependent survival of immature individuals, the maturation delay trade-off had some interesting effects on population stability. For some parameters, we observed that if a small maturation delay prevented the existence of an interior equilibrium, increasing the maturation delay was able to promote the existence of a stable interior equilibrium. A further increase in the maturation delay could however destabilize this interior equilibrium through either a period doubling or Neimark--Sacker bifurcation. This oscillatory behavior could be removed with a further increase in the maturation delay with extinction being the final stage. Unexpected was the dependence of the bifurcation type on the evenness of the maturation delay parameter $\tau$. An even value of $\tau$ promoted the destabilization of the positive equilibrium, as the egg productivity value $a$ increased, a Neimark--Sacker bifurcation, while an odd value of $\tau$ resulted in a period-doubling bifurcation. 

Our maturation delay augments the underlying oscillatory dynamics inherent to the Ricker framework. Classic work examining the Ricker model shows how chaos occurs through increasing the intrinsic growth rate, $r$ \cite{mayBifurcationsDynamicComplexity1976}. Our model splits $r$ into birth of new immature individuals and deaths of mature and immature individuals. Our maturation delay influences the births and deaths of immature individuals and so equivalently changes $r$ of the original Ricker model. Initially, increasing the maturation delay increases the fecundity and effectively increases $r$ of the original Ricker model, creating oscillations. Of course, when the maturation delay is too long and the immature losses outweigh the fecundity benefits, then effectively, we have returned to a small $r$ from the original Ricker model where oscillations do not happen (or extinction occurs). Two key concepts in theoretical ecology to interpret the outcome of increasing maturation delay are mean-driven instability and variance-driven instability (equivalently persistence and bounded from zero in the mathematics literature) \cite{gellner_duality_2016}. Mean-driven instability is when a population's mean is close to zero and then by stochasticity can easily go extinct. Variance-driven instability is when even though the mean is much larger than zero, the oscillations drive the population again close to zero, leading to increased chances of extinction by stochasticity. Our maturation delay initially reduces mean-driven instability (pushing the mean further from zero), then increases variance-driven instability (inducing oscillations), before reducing variance-driven instability (removing oscillations), and increases mean-driven instability (reducing the mean). 

The intraspecific density-dependence of the mature and immature individuals had a significant effect on the presence or absence of oscillatory dynamics in the Ricker-Ricker model. Depending on the $\tau$ and $a$, weakening intraspecific density-dependence regularly removes oscillations for $C$ and induces oscillations for $\beta$. There are of course nuances to this where for especially small values of $C$, oscillations can be induced again. Following the logic of the previous paragraph, we could expect that the weakening of intraspecific density-dependence to act similarly to increasing r in the original Ricker model (effectively reducing the probability of losing individuals). Although, this could be the case for mature density-dependence, weakening immature density-dependence largely has the opposite effect. Therefore, these patterns point to potential important differences in how mature versus immature density-dependence act to affect population dynamics, and the importance of the interaction between the mature and immature population dynamics. Indeed, \cite{reesGrowthReproductionPopulation1989} outline  a major difference between animals and plants in terms of their population dynamics. Animals generally have a size or age threshold that must be reached to produce new offspring whereas this threshold rarely exists in plants. \cite{reesGrowthReproductionPopulation1989} then postulate that this threshold difference is a major cause of oscillatory population dynamics in animals and relatively stable population dynamics in plants. Interestingly, many modeling studies that examine maturation delays assume stable size or age distributions and then examine the population growth rate, $r$ (e.g \cite{stearnsEvolutionPhenotypicPlasticity1986} \& \cite{takadaOptimalSizeMaturity1997}). Although useful from a population ecological-evolutionary perspective, this form of modeling misses out on the dynamic interactions of mature and immature populations. Consequently, when examining the population dynamics of species due to maturation delays, incorporating the dynamics of mature and immature individuals is critical.

Our delay model sets up a useful phenomenological framework to test the interactions of trade-offs in parent survival, offspring survival, and reproductive investment. Fundamentally, life history trade-offs occur through energy allocation between the biological needs of growth, somatic maintenance, defense, reserves, and reproduction \cite{zeraPhysiologyLifeHistory2001a}. Because a maturation delay is effectively an individual allocating energy to growth, defense, or reserves over reproduction, parent and immature survival and future reproduction can be affected by a maturation delay \cite{zeraPhysiologyLifeHistory2001a}. Furthermore, different life history traits often co-vary together with maturation delay even within a species \cite{purchaseSexspecificCovariationLifehistory2005}. Our model's separate functions for parent survival, offspring survival, and reproduction while being able to manipulate the maturation delay allows for comparative modeling across multiple axes of life history trade-offs. All of the potential combinations tested using our modeling framework could be empirically matched with data to identify how parent survival, offspring survival, and reproductive investment co-vary with maturation delay either between species or within species. From this, the dynamical predictions of this co-variation could be tested against different species' population data.

Adding a maturation delay that accounted for immature losses was an initial important update to classic difference equations. We further extended this logic by adding a fecundity benefit to the maturation delay because maturation delays are ubiquitous in nature. Both the Beverton--Holt and Ricker modeling frameworks exhibited an optimal maturation delay where the fecundity benefits matched the immature losses from the maturation delay. Similar to \cite{StWodelay2, StWodelay3} that added a maturation delay (but without the fecundity benefit), long maturation delays still caused population extinction. The maturation delay also interacted with the inherent oscillatory nature of the Ricker model, where increasing the maturation delay initiated oscillations, but with a further increase in maturation delay, oscillations were removed. Although this initial model is limited solely to changing the fecundity with maturation delay, the model set-up of separating the mature and immature cohorts is primed to explore many combinations of life history trade-offs. Taken together, our extended model incorporates an important life history trade-off critical to modeling species populations.

Extensions of our model may include other maturation delay benefits beyond fecundity. For example, the per time unit survival rates ($\alpha$ \& $D$) that we chose as constants could be delay dependent or even be density-dependent. A consequence of these constant model parameters is that a population with a shorter maturation delay would grow to a smaller size but experience the same per time unit survival rates as a population with longer maturation delay that grows to a larger size. Using the energy trade-off framework, a justification for this assumption could be that larger animals need more energy for maintenance and growth relative to survival, in turn maintaining the per time unit survival rates at all sizes. However, total energy intake increases with size and generally total energy usage per unit of body mass decreases \cite{marshallGlobalSynthesisOffspring2018}. Furthermore, size often positively impacts survival rates either through reducing predation and other natural mortality  \cite{gislasonSizeGrowthTemperature2010,trexlerEffectsHabitatBody1992} or by increasing lifespan \cite{speakmanBodySizeEnergy2005}. Consequently, if we think about our per time unit survival rates as average per capita survival across the population, the average per time unit survival rate should increase with a longer maturation delay and larger sizes. Future work could include these size-dependent survival rates alongside size-dependent fecundity.   


\appendix
\renewcommand{\thesection}{\Alph{section}}
\renewcommand{\thesubsection}{\Alph{section}.\arabic{subsection}}

\section{Proof of Theorem~\ref{thm:NstarLAS1}}\label{A:NstarLAS1}.

\begin{proof}
Assume $p(0)>1-m(\tau)$, then, by assumption (H) given in \eqref{eq:H}, there exists $N^*_+>0$ such that $p(N^*_+)=1-m(\tau)$. By \eqref{Clargen}, $N^*_+$ is LAS if 
\begin{equation}\label{LASXstar}
    \left|p(N^*_+)+N^*_+p'(N_+^*)\right|<1-m(\tau)=p(N^*_+),
\end{equation}
i.e., 
\begin{equation*}
-p(N^*_+)<p(N^*_+)+N^*_+p'(N_+^*)<p(N^*_+).
\end{equation*}
By assumption (H) given in \eqref{eq:H}, $p$ decreases so that the second inequality holds automatically. Thus, it suffices to show that 
$$-p(N^*_+)<p(N^*_+)+N^*_+p'(N^*_+).$$
Note that since $p(N^*_+)=1-m(\tau)$ and, by assumption (H) given in \eqref{eq:H}, $p$ is injective, $N^*_+=p^{-1}(1-m(\tau))$. Since $m(\tau)\in (0,1)$, $N^*_+$ is bounded. That is, there exists $M>0$ such that $N^*_+\leq M$. 

Furthermore, by assumption, $p'(z)$ is continuous. Define  $\underline{dp}:=\inf_{z\in [0, M]} p'(z)$, then $\underline{dp}\in (-\infty, 0)$. Thus, it suffices to show that
\begin{equation}\label{eq:35star}
-p(N^*_+)<p(N^*_+)-N^*_+|\underline{dp}|.
\end{equation}
Since $N^*_+$ can only exist if $p(0)>1-m(\tau)$ because for $p(0)\leq 1-m(\tau)$, $N^*_0$ is the only nonnegative equilibrium,  the following holds by the continuity of the equilibrium $N^*_+$:\\
For all $\epsilon>0$, there exists $\delta>0$ such that if $ 1-m(\tau) \leq p(0)\leq 1-m(\tau)+\delta$, $0\leq N^*_+\leq \epsilon$. 
Thus, choosing $\epsilon=\frac{p(1)}{|\underline{dp}|}>0$, then there exists $\delta>0$ such that $N^*_+\in (0,\min\{1,\epsilon\})$ is an equilibrium and satisfies \eqref{eq:35star} and therefore also the stability condition \eqref{LASXstar}. 
\end{proof}

\section{Proof of Theorem~\ref{Nstarex}}\label{Thm:Nstarex}

\begin{proof}
Let $P:=\ln(1/\overline{p})>0$ and $x:=e^{-(\tau+1)}<1$. Then, we can rewrite \eqref{condNstar} as
\begin{equation*}
S(x):=ax^P
- b x^{P+K}-\frac{\alpha}{1+\alpha}>0.
\end{equation*}
Since $\lim_{x\to 0^+}S(x)<0$,  the desired inequality does not hold as $\tau\to \infty$, implying that there exists $\tau_c$ such that $N^*(\tau)$ does not exist for all $\tau>\tau_c$. Furthermore,
$$S'(x)=aPx^{P-1}-b(P+K)x^{P+K-1}=x^{P-1}(aP-b(P+K)x^K)$$
which has exactly one positive (real) critical point at  
$$x_c=\left(\frac{aP}{b(P+K)}\right)^{\frac{1}{K}}>0.$$
Since 
$$S''(x)=aP(P-1)x^{P-2}-b(P+K)(P+K-1)x^{P+K-2}=x^{P-2}(aP(P-1)-b(P+K)(P+K-1)x^K)$$
and 
$$S''(x_c)=x_c^{P-2}\left(aP(P-1)-b(P+K)(P+K-1)\left(\frac{aP}{b(P+K)}\right)\right)
=x_c^{P-2}(-aPK)<0,
$$ 
$x=x_c$ determines the location of the global maximum of $S(x)$ with the value
$$S(x_c)= \left(\frac{aP}{b(P+K)}\right)^{\frac{P}{K}}\left(a-b\left(\frac{aP}{b(P+K)}\right)\right)-\frac{\alpha}{1+\alpha}=
 \left(\frac{aP}{b(P+K)}\right)^{\frac{P}{K}}\frac{aK}{P+K}-\frac{\alpha}{1+\alpha}.
$$
Thus, if $S(x_c)>0$, then there exists $\tau_1<\tau_2$ such that $N^*_+(\tau)$  exists (biologically relevant) for $\tau\in (\tau_1,\tau_2)$. If $S(e^{-1})>0$, then $\tau_1<0$, else, if $S(e^{-1})< 0$, then $\tau_1> 0$.

\end{proof}

\section{Proof of Theorem ~\ref{ThmexNstar}}\label{proofThmexNstar}. 

\begin{proof}
From \eqref{T1C1Nstar}, we calculate
\begin{equation*}
    \frac{\partial N^*_+}{\partial \tau}
    =\frac{\overline{p}^{\tau+1} (b K+ (ae^{K(\tau+1)} -b)\ln(\overline{p}))}{\beta (b p^{
    1 + \tau} e^{-K(\tau+1)} +  (1 - a p^{1 + \tau}))^2}.
    \end{equation*}
Since the sign depends on the sign of $b K+ (ae^{K(\tau+1)} -b)\ln(\overline{p})$ and 
$$b K+ (ae^{K(\tau+1)} -b)\ln(\overline{p})=0\qquad \mbox{for} \quad \tau=\hat{\tau}:=\frac{1}{K}\ln\left( \frac{b\ln(1/\overline{p})+bK}{ae^{K}\ln(1/\overline{p})}  \right),$$
we immediately have that 
\begin{equation*}
  \frac{\partial N^*_+}{\partial \tau}\qquad    
\begin{cases} 
<0 & \tau>\widehat{\tau}\\
=0 & \tau=\widehat{\tau}\\
>0& \tau<\widehat{\tau}
    \end{cases}
\end{equation*}
completing the proof.
\end{proof}

\section{Proof of Theorem~\ref{thm:NstarLAS}}\label{A:NstarLAS}

\begin{proof}
By Theorem~\ref{ThmexNstar}, if $m(\tau)>\frac{\alpha}{1+\alpha}$,  $N^*_+>0$. In that case, by \eqref{LASClark}, $N^*_+$ is LAS if 
$$(1+\alpha +\beta N^*_+)^2>\frac{1+\alpha}{1-m(\tau)}\quad \iff \quad N^*_+>\frac{\sqrt{1+\alpha}}{\beta\sqrt{1-m(\tau)}}-\frac{1+\alpha}{\beta}=\frac{\sqrt{1+\alpha}}{\sqrt{1-m(\tau)}}\frac{(1-\sqrt{(1+\alpha)(1-m(\tau))})}{\beta}.$$
By \eqref{T1C1Nstar}, this implies that  
$$N^*_+=\frac{(1+\alpha) m(\tau)-\alpha}{\beta(1 -m(\tau))}>\frac{\sqrt{1+\alpha}}{\sqrt{1-m(\tau)}}\frac{(1-\sqrt{(1+\alpha)(1-m(\tau))})}{\beta},$$
i.e., 
$$(1+\alpha) m(\tau)-\alpha>\sqrt{(1+\alpha)(1 -m(\tau))}(1-\sqrt{(1+\alpha)(1-m(\tau))})=\sqrt{(1+\alpha)(1 -m(\tau))}-(1+\alpha)(1-m(\tau)),$$
i.e., 
$$1>\sqrt{(1+\alpha)(1 -m(\tau))}\qquad \iff \qquad m(\tau)>\frac{\alpha}{1+\alpha}$$
Note that, by \eqref{condN0}, the condition $m(\tau)>\frac{\alpha}{1+\alpha}$ also implies that $N^*_0$ is unstable. 

\end{proof}

\section{Theorem~1.15 in  \cite{Ladas_2004}}\label{A:ThmLadas}

\begin{theorem}[See \protect{\cite[Theorem 1.15]{Ladas_2004}}]\label{ThmLadas}
Let $g:[a,b]^{k+1} \to [a,b]$ be a continuous function, where $k$ is a positive integer and $[a,b]$ is an interval of real numbers. Consider 
$$x_{n+1} = g(x_n,x_{n-1},\ldots, x_{n-k}),\quad \quad n=0,1,\ldots$$
Assume that $g$ satisfies the following conditions: 
\begin{enumerate}
    \item For each integer $i$ with $1\leq i\leq k+1$, the function $g(z_1,z_2,\ldots, z_{k+1})$ is weakly monotonic in $z_i$ for fixed $z_j$, $j\neq i$.
    \item If $(r,R)$ $\in \mathbb{R}^2$ is a solution of the system
$$    r=g(r_1,r_2,\ldots, r_{k+1}), \qquad \qquad 
    R=g(R_1,R_2,\ldots, R_{k+1}),$$
   where for each $i=1,2,\ldots, k+1$, we set
    $$R_i = \begin{cases} R & \mbox{ if g is non-decreasing in }z_i\\
    r & \mbox{ if g is non-increasing in }z_i\end{cases}$$
    and 
    $$r_i = \begin{cases} r & \mbox{ if g is non-decreasing in }z_i\\
    R & \mbox{ if g is non-increasing in }z_i\end{cases}$$
\end{enumerate}
and $r=R$, then there exists exactly one equilibrium $\bar{x}$ and every solution converges to $\bar{x}$.
\end{theorem}

Note that we call {\it weakly monotonic}, in the statement of the Theorem 1.15 in \cite{Ladas_2004}, {\it component-wise monotone}.

\section{Proof of Theorem~\ref{thm:NstarGAS}}\label{A:NstarGAS}

\begin{proof}
If $(1+\alpha)m(\tau)>\alpha$, $N^*_+>0$ exists (biologically relevant). 
Since $\frac{\partial F(u,v)}{\partial u}=Q(u)>0$ and $\frac{\partial F(u,v)}{\partial v}=m(\tau)>0$, the map $F$ is component-wise monotone (increasing) in each variable. 
To show that there exists $0<a_r<b_r$ such that the set $[a_r,b_r]^{\tau+1}$ is positively invariant, we first show that there exists $a_r>0$ such that 
$F(a_r,a_r)\geq a_r$. We have
\begin{align*}
F(a,a)&=\left(\frac{1}{1+\alpha+\beta a} + m(\tau)\right)a\geq a \qquad \iff \qquad 
\frac{1}{1+\alpha+\beta a} + m(\tau)\geq 1,
\end{align*}
i.e., 
$$1+m(\tau)(1+\alpha)+m(\tau)\beta a\geq 1+\alpha+\beta a,$$
i.e., 
$$a<\frac{m(\tau)(1+\alpha)-\alpha}{\beta (1-m(\tau))}=:a_\ell.$$
Since $m(\tau)(1+\alpha)>\alpha$, the numerator is positive and since, by assumption (A) given in \eqref{eq:A}, $m(\tau)<1$, $a_r$ can be chosen arbitrarily small such that $a_r<a_\ell$. 

To show that there exists $b_r>a_r$ such that $F(b_r,b_r)\leq b_r$, note that for $b>0$,  
$$F(b,b)=\left(\frac{1}{1+\alpha+\beta b} + m(\tau)\right)b\leq b,$$
if and only if 
$$m(\tau)(1+\alpha)+\beta (\tau)  b m(\tau)\leq \alpha+\beta b,$$
i.e, 
$$b\geq \frac{m(\tau)(1+\alpha)-\alpha}{\beta (1-m(\tau))}=:b_\ell,$$
which is positive because $0<m(\tau)<1$ and, by assumption, $(1+\alpha)m(\tau)>\alpha$ so that also the numerator is positive. Thus, we choose $b_r$ arbitrarily large but such that $b_r\geq \max\{b_\ell, a_r\}$. 
Then, $[a_r, b_r]^{\tau+1}$ is a positively invariant set. Since $a_r$ can be arbitrarily small and $b_r$ can be arbitrarily big, we can guarantee that $N^*_+\in [a_r,b_r]$ and $N^*_0\notin [a_r, b_r]$. 

To apply Theorem~\ref{ThmLadas}, it remains to show that
$$r=F(r,r), \qquad R=F(R,R)$$
has only the trivial solution in $[a_r,b_r]$ given by  $N^*_+$. Since $x=F(x,x)$ is equivalent to the equilibrium condition and we chose $a_r, b_r$ such that the only equilibrium in $[a_r,b_r]$ is $N^*_+$, we have $r=R=N^*_+$. By Theorem~\ref{ThmLadas}, every solution that enters $[a_r,b_r]$ converges to $N^*_+$ and since it is locally asymptotically stable, as established in Theorem~\ref{thm:NstarLAS},  $N^*_+$ is indeed globally asymptotically stable for all solutions that enter $[a_r,b_r]$. Since, $a_r$ can be  chosen arbitrarily small and $b_r$ arbitrarily big, we can ensure that solutions to positive initial conditions are automatically in $[a_r, b_r]$, completing the proof. 

\end{proof}

\section{Proof of Theorem~\ref{thm:N0LAS}}\label{A:thmN0LAS}

\begin{proof}

By \eqref{LASClark}, the trivial equilibrium $N^*_0=0$ is locally asymptotically stable if 
$$1>
\frac{1}{1+\alpha}+\left(a - b e^{-K(\tau+1)}\right)
\overline{p}^{\tau+1}$$
i.e., 
\begin{equation}\label{condN0}
\left(a - b e^{-K(\tau+1)}\right)
\overline{p}^{\tau+1}<\frac{\alpha}{1+\alpha}.
\end{equation}
Note that the left-hand side is the expression of $m(\tau)$ so that the sufficient condition for the local asymptotic stability of $N^*_0$ reads as $m(\tau)<\frac{\alpha}{1+\alpha}$. 
Since 
$$m'(\tau)=\underbrace{e^{-K(\tau+1)}\overline{p}^{ 1 + \tau}}_{>0} (\underbrace{b K+b\ln(1/\overline{p})}_{>0} -\underbrace{a\ln(1/\overline{p}) e^{K (1 + \tau)}}_{>0}),$$
we note that for $P:=\ln(1/\overline{p})$, 
$$m'(\tau)<0 \qquad \iff \qquad  \tau> \frac{1}{K}\ln\left(\frac{b(P+K)}{aP}\right) -1.$$
Since $\lim_{\tau\to \infty}m(\tau)=0$, this implies the existence of $\tau_c> \frac{1}{K}\ln\left(\frac{b(P+K)}{aP}\right) -1$ such that $N^*_0$ is locally asymptotically stable for all $\tau>\tau_c$, completing the proof.
\end{proof}

\section{Proof of Theorem~\ref{thm:GASNstar}}\label{A:GASNstarR}


\begin{proof}
    Assume $1-e^{-\alpha}< m(\tau)$, then there exists a positive equilibrium. 
By the equilibrium expression \eqref{eq:q}, 
$$\beta N^*_+<2\qquad \iff \qquad
m(\tau)<1-e^{-(\alpha+2)}.$$ 
By  \eqref{LASClark}, it suffices to show that $$\left| e^{-\alpha-\beta N^*_+}(1-\beta N^*_+) \right|+\left|m(\tau)\right|<1.$$
    At an equilibrium, $m(\tau)=1-e^{-\alpha -\beta N^*_+}$. Thus, if $\beta N^*_+\leq 1$, we have 
\begin{align*}
    \left| e^{-\alpha-\beta N^*_+}(1-\beta N^*_+) \right|+\left|m(\tau)\right|=
    e^{-\alpha-\beta N^*_+}(1-\beta N^*_+) +1-e^{-\alpha-\beta N^*_+}=1-e^{-\alpha-\beta N^*_+} \beta N^*_+<1.
\end{align*}
Instead if $\beta N^*_+>1$, then since $\beta N^*_+< 2$,  
\begin{align*}
    \left| e^{-\alpha-\beta N^*_+}(1-\beta N^*_+) \right|+\left|m(\tau)\right|=
    e^{-\alpha-\beta N^*_+}(\beta N^*_+-1) +1-e^{-\alpha-\beta N^*_+}=
    (\beta N^*_+-2)e^{-\alpha-\beta N^*_+}+1<1,
\end{align*}
This confirms that $N^*_+$ is locally asymptotically stable.  

Instead, if $m(\tau)>1-e^{-\alpha-2}$, then $\beta N^*_+>2$ and the stability condition is violated and the equilibrium is no longer stable, see also the proof in \cite[Theorem 3.4]{streipertAlternativeDelayedPopulation2021}.
\end{proof}

\section{Conjecture of global asymptotic stability for Ricker-Constant model}\label{A:GASNstarR_figs}

\begin{figure}[H]
    \centering
    \includegraphics[width=0.8\linewidth]{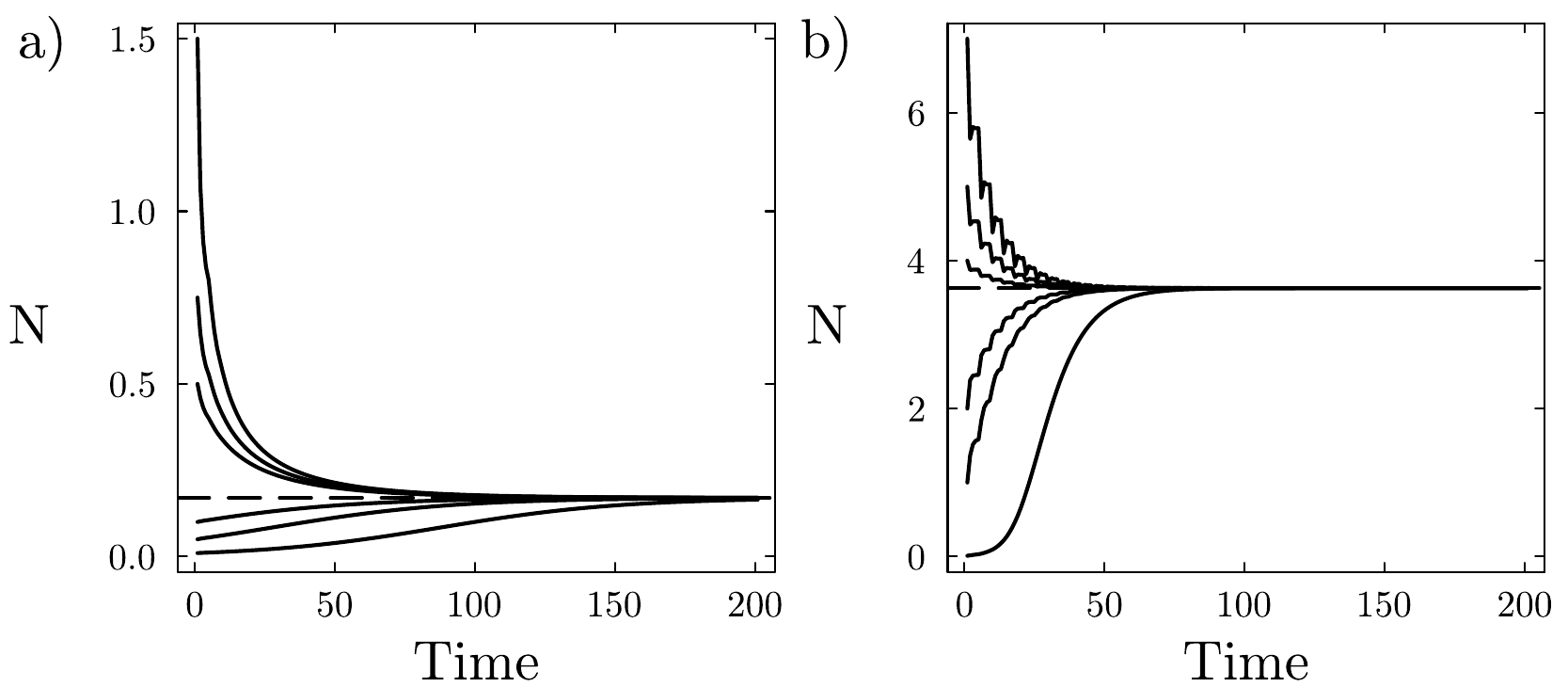}
    \caption{Ricker-Constant time series where $1-e^{-\alpha}< m(\tau)<1-e^{-\alpha-2}$. Dashed lines are $N^*_+$. Solid lines are model trajectories with different initial conditions. a) $a=13.0$, $p=0.35$, $\tau=3.0$, $\alpha=0.1$, $\beta=0.3$, $b=200$, $K=1.0$. Initial conditions for each line are $(N_{-3}, N_{-2}, N_{-1}, N_0)=(0.01,0.01,0.01,0.01)$, $(0.05,0.05,0.05,0.05)$, $(0.1,0.1,0.1,0.1)$, $(0.5,0.5,0.5,0.5)$, $(0.75,0.75,0.75,0.75)$, $(1.5,1.5,1.5,1.5)$ respectively. b) $a=50.0$, $p=0.35$, $\tau=3.0$, $\alpha=0.1$, $\beta=0.3$, $b=200$, $K=1.0$. Initial conditions for each line are $(N_{-3}, N_{-2}, N_{-1}, N_0)=(0.01,0.01,0.01,0.01)$, $(1.0,1.0,1.0,1.0)$, $(2.0,2.0,2.0,2.0)$, $(4.0, 4.0,4.0,4.0)$, $(5.0,5.0,5.0,5.0)$, $(7.0,7.0,7.0,7.0)$ respectively.}
    \label{fig:RC_stable}
\end{figure}

\begin{figure}[H]
    \centering
    \includegraphics[width=0.8\linewidth]{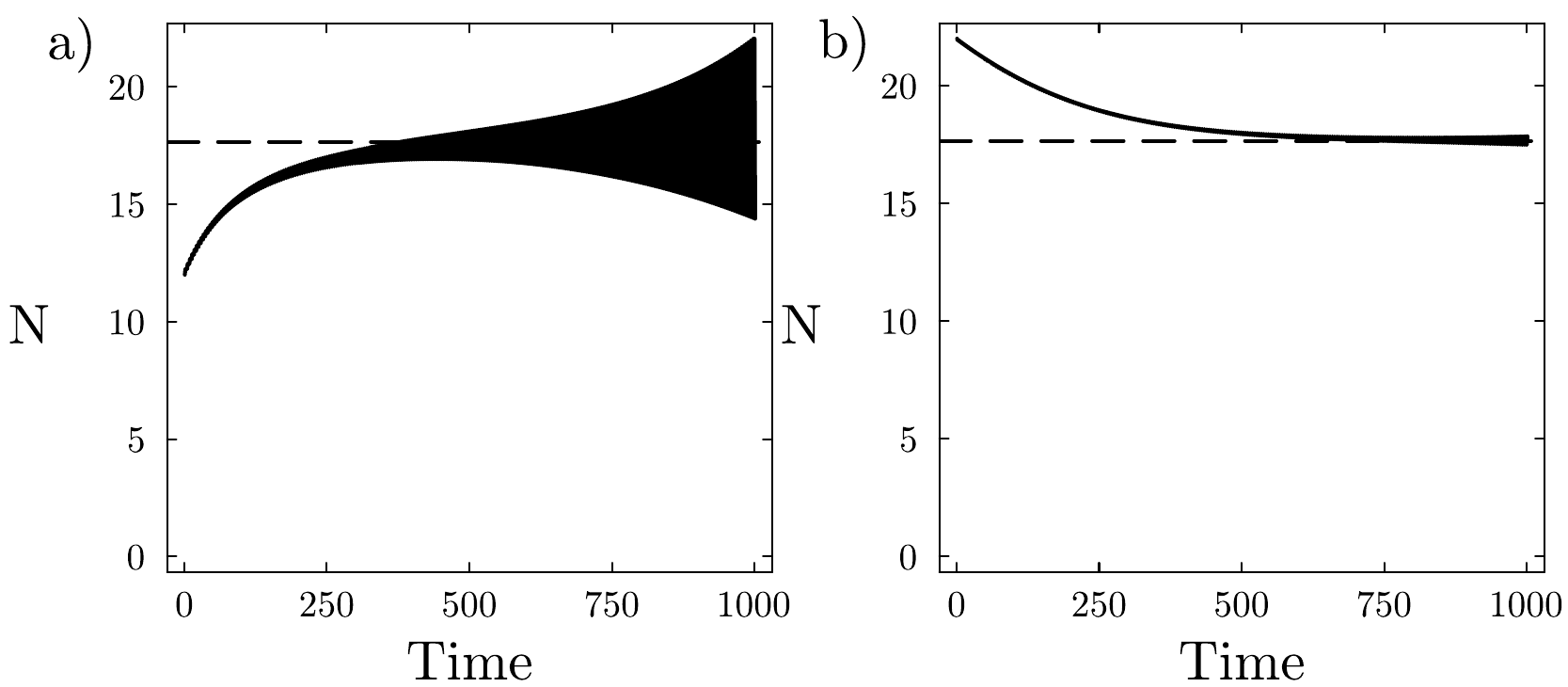}
    \caption{Ricker-Constant time series where $m(\tau)>1-e^{-\alpha-2}$. Dashed lines are $N^*_+$. Solid lines are model trajectories. For both a) and b) $a=70.0$, $p=0.35$, $\tau=3.0$, $\alpha=0.1$, $\beta=0.3$, $b=200$, $K=1.0$. a) The initial conditions for the trajectory is $(N_{-3}, N_{-2}, N_{-1}, N_0)=(12.0, 12.0, 12.0, 12.0)$. b) The initial conditions for the trajectory is $(N_{-3}, N_{-2}, N_{-1}, N_0)=(22.0, 22.0, 22.0, 22.0)$} 
    \label{fig:RC_unstable}
\end{figure}

\section{Proof of Theorem~\ref{thm:PeriodD}}\label{A:thmPeriodD}

%


\begin{proof}
If $\tau\in \mathbb{N}$ for some $n\in \mathbb{N}$, then for $\Vec{v}=(N_t, N_{t-1}, N_{t-2}, \ldots, N_{t-\tau})$, the Jacobian of \eqref{EqT1R} at $N^*>0$ is 
\begin{align*}
J&=
\begin{bmatrix}
   A & 0 & 0 & \ldots & 0 &0 & B \\
1 & 0 & 0 &  \ldots & 0 & 0 & 0\\
0 & 1 & 0 &  \ldots & 0 & 0 & 0\\
\vdots & \vdots & \vdots & \vdots \vdots \vdots & \vdots & \vdots & \vdots \\
0 & 0 & 0 &  \ldots & 0 & 1 & 0\\
\end{bmatrix}\in \mathbb{R}^{(\tau+1)\times (\tau+1)},
\end{align*}
where 
$$A= e^{-\alpha-\beta N^*}(1-\beta  N^*) =(1-\left(a - b e^{-(\tau+1)K}\right)
\overline{p}^{\tau+1})(1-\beta  N^*),$$ 
$B=\left(a - b e^{-(\tau+1)K}\right)
\overline{p}^{\tau+1}$ and \begin{equation}\label{eq:Nstarforproof}
N^*=\frac{1}{\beta}\left\{-\ln\left(1-\left(a - b e^{-(\tau+1)K}\right)
\overline{p}^{\tau+1}\right)-\alpha\right\}.
\end{equation}

The characteristic equation of the Jacobian at $N^*$ is then 
\begin{equation}\label{charact}
\lambda^{\tau+1}-A\lambda^\tau-B=
\lambda^{\tau+1} - (1-\left(a - b e^{-(\tau+1)K}\right)
\overline{p}^{\tau+1})(1-\beta  N^*) \lambda^{\tau} -\left(a - b e^{-(\tau+1)K}\right)
\overline{p}^{\tau+1}=0.
\end{equation}

Assume first that $\tau=2n+1$ for some $n\in \mathbb{N}$. Suppose $\lambda=-1$ is an eigenvalue, then
$$1 + (1-\left(a - b e^{-2(n+1)K}\right)
\overline{p}^{2(n+1)})(1-\beta  N^*) -\left(a - b e^{-2(n+1)K}\right)
\overline{p}^{2(n+1)}=0,$$
i.e., by \eqref{eq:Nstarforproof}, 
$$1 + (1-\left(a - b e^{-2(n+1)K}\right)
\overline{p}^{2(n+1)})(1+\ln\left(1-\left(a - b e^{-2(n+1)K}\right)
\overline{p}^{2(n+1)}\right)+\alpha) -\left(a - b e^{-2(n+1)K}\right)
\overline{p}^{2(n+1)}=0,$$
i.e., 
$$ (1-\left(a - b e^{-2(n+1)K}\right)
\overline{p}^{2(n+1)})(2+\ln\left(1-\left(a - b e^{-2(n+1)K}\right)
\overline{p}^{2(n+1)}\right)+\alpha) =0,$$
i.e., 
$$ \ln\left(1-\left(a - b e^{-2(n+1)K}\right)
\overline{p}^{2(n+1)}\right) =-2-\alpha,$$
which, after solving for $a$ implies that 
$$a=a^*:=\frac{1- e^{-2-\alpha}}{\overline{p}^{2(n+1)}}+ b e^{-2(n+1)K}>0.$$
Furthermore, since we require $m(\tau)=a-be^{-(\tau+1)K}\overline{p}^{\tau+1}<1$, we require 
$$(a^*-be^{-2(n+1)K})\overline{p}^{2(n+1)}<1\qquad \iff \qquad a^*<\frac{1}{\overline{p}^{2(n+1)}} +be^{-2(n+1)K},$$
which is certainly satisfied. Hence, there exists $a=a^*$ such that $\lambda=-1$ is a root of \eqref{charact}.

Assume now that $\tau=2n$ for some $n\in \mathbb{N}$. Suppose $\lambda=-1$ is an eigenvalue, then \eqref{charact} implies
$$-1 - (1-\left(a - b e^{-(2n+1)K}\right)
\overline{p}^{2n+1})(1-\beta  N^*)  -\left(a - b e^{-(2n+1)K}\right)
\overline{p}^{2n+1}=0,$$
which, by \eqref{eq:Nstarforproof},  can be written as
$$-1 - (1-\left(a - b e^{-(2n+1)K}\right)
\overline{p}^{2n+1})(1+ \ln\left(1-\left(a - b e^{-(2n+1)K}\right)
\overline{p}^{2n+1}\right)+\alpha)  -\left(a - b e^{-(2n+1)K}\right)
\overline{p}^{2n+1}=0,$$
i.e., $$ (-1+\left(a - b e^{-(2n+1)K}\right)
\overline{p}^{2n+1})\left\{1+ \ln\left(1-\left(a - b e^{-(2n+1)K}\right)
\overline{p}^{2n+1}\right)+\alpha\right\} =1+\left(a - b e^{-(2n+1)K}\right)
\overline{p}^{2n+1}.$$

Remember that $m=m(\tau)=(a-be^{-(2n+1)K})\overline{p}^{2n+1}$ and $0<m(\tau)<1$. Then, the above equality can be written as
\begin{equation}\label{eq:psineed}
 \Psi(m):=(-1+m)(1+ \ln\left(1-m\right)+\alpha) =1+m.
 \end{equation}
By assumption, we require $1<1+m<2$ and since $\Psi(0)=-(1+\alpha)<0<1+m$, we have, by l'Hopital, 
$$\lim_{m\to 1^-}\Psi(m)=\lim_{m\to 1^-} \frac{1+ \ln\left(1-m\right)+\alpha}{\frac{1}{m-1}}
= \lim_{m\to 1^-} \frac{\frac{-1}{1-m}}{\frac{-1}{(m-1)^2}}=\lim_{m\to 1^-} (1-m) =0<2=\lim_{m\to 1^-}1+m.
$$
Furthermore,  $\frac{\partial \Psi(m)}{\partial m}=2+\alpha+\ln(1-m)=0$ if and only if $m=m_{max}=1-e^{-2-\alpha} \in (0,1)$ and since 
$$\Psi''(m_{max})=\frac{-1}{(1-m_{max})}=\frac{-1}{e^{-2-\alpha}}<0,$$
the maximum occurs at $m=m_{max}$. 
Thus, 
\begin{equation*}
\begin{split}
    \Psi(m)&\leq \Psi(m_{max})=(-1+(1-e^{-2-\alpha}))(1+\ln(1-(1-e^{-2-\alpha}))+\alpha)=    -e^{-2-\alpha}(1+\ln(e^{-2-\alpha})+\alpha)\\
    &=    -e^{-2-\alpha}(-1)=
    e^{-2-\alpha}<1<1+m,
\end{split}
\end{equation*}
disproving the possibility for the existence of $m$ such that  \eqref{eq:psineed} holds. Hence, if $\tau$ is even, $\lambda=-1$ can't be a root of \eqref{charact}.
\end{proof}

\section{Proof of Theorem~\ref{thm:PeriodNS}}\label{A:thmPeriodNS}


\begin{proof}
Let $\tau=2n$ for some $n\in \mathbb{N}$. Note that by the definition of $m(\tau)$ in \eqref{eq:mtau}, if we can find $m(\tau)\in (0,1)$ such that a Neimark--Sacker bifurcation occurs, then we can also find a corresponding value of $a$. 
    By \eqref{charact}, the characteristic equation of the Jacobian evaluated at $N^*$ is given by 
    $$\lambda^{2n+1} - (1-\left(a - b e^{-(2n+1)K}\right)
\overline{p}^{2n+1})(1-\beta  N^*) \lambda^{2n} -\left(a - b e^{-(2n+1)K}\right)
\overline{p}^{2n+1}=0,$$
i.e., with $m(\tau)$ defined in \eqref{eq:mtau}, 
$$\lambda^{2n+1} - (1-m(\tau))(1-\beta  N^*) \lambda^{2n} -m(\tau)=0.$$
By \eqref{eq:Nstarforproof}, this can be written as 
$$\lambda^{2n+1} - (1-m(\tau))(1+\alpha + \ln(1-m(\tau)) \lambda^{2n} -m(\tau)=0.$$
By Theorem~\ref{thm:PeriodD}, there does not exist a feasible $a$ such that there is a period doubling bifurcation. If there exists a transcritical bifurcation, then $\lambda=1$ is a root, that is, 
\begin{align*}
0&=1- 
(1+\alpha + \ln(1-m(\tau))  
+m(\tau)(1+\alpha + \ln(1-m(\tau))  
-m(\tau)\\&=- 
(\alpha + \ln(1-m(\tau))  
+m(\tau)(\alpha + \ln(1-m(\tau)) = (m(\tau)-1)(\alpha + \ln(1-m(\tau)).
\end{align*}
Since $m(\tau)<1$, this can only be satisfied if $\alpha=-\ln(1-m(\tau))$, or, equivalently, $m(\tau)=1-e^{-\alpha}\in (0,1)$. Note that this is exactly the threshold for the existence of the positive equilibrium $N^*$. By Lemma~\ref{cor:N0GASR}, we may therefore consider $m(\tau)>1-e^{-\alpha}$. Thus, no transcritical bifurcation can occur for feasible values of $a$. 
Hence, if one of the Jury conditions is violated for $m(\tau)\in (0,1)$ with $m(\tau)>1-e^{-\alpha}$, then it must be due to a Neimark--Sacker bifurcation.

By the above, the characteristic polynomial at $N^*>0$ is of the form 
$$P(\lambda):=\lambda^{\tau+1}+\sum_{i=0}^{\tau+1}a_i\lambda^{\tau+1-i}\qquad \qquad a_i=\begin{cases}
-m(\tau) & i=\tau+1\\
0 & 2\leq i \leq \tau\\
-(1-m(\tau))(1+\alpha+\ln(1-m(\tau)))& i=1
\end{cases}
$$
Based on \cite[p.~59]{Keshet}, one of the Jury conditions is violated if $|b_n|\leq |b_1|$, where $n=\tau+1$. Since  
$$b_n=1-a_n^2=1-m^2(\tau)>0$$
and 
$$b_1=a_{n-1}-a_na_1=-m(\tau)(1-m(\tau))(1+\alpha+\ln(1-m(\tau)))).$$
Thus, the condition is violated if 
$$1-m^2(\tau)\leq m(\tau)(1-m(\tau))|1+\alpha+\ln(1-m(\tau))|$$
Note that we already assume that $1-e^{-\alpha}<m(\tau)<1$ to guarantee the existence of $N^*$ so that $1+\alpha+\ln(1-m(\tau))<0$. If there exists $m(\tau)$ such that i) $1+\alpha+\ln(1-m(\tau))<0$, ii) $1-e^{-1-\alpha}<m(\tau)<1$,
and iii)
$1-m^2(\tau)\leq -m(\tau)(1-m(\tau))(1+\alpha+\ln(1-m(\tau)))$, then the Jury condition is violated. Note that ii) implies i). 

Condition iii) can be equivalently expressed as
$$1+m(\tau)\leq -m(\tau)(1+\alpha+\ln(1-m(\tau))),$$
i.e., 
$$\frac{-1}{m(\tau)}-2-\alpha\geq \ln(1-m(\tau)).$$
If $m(\tau)>1-e^{-\alpha}$, then it suffices to show that we can find $m(\tau)$ such that 
$$\frac{-1}{1-e^{-\alpha}}-2-\alpha\geq \ln(1-m(\tau)),$$
i.e., 
$$m(\tau)\geq 1-e^{\frac{-1}{1-e^{-\alpha}}-2-\alpha}>1-e^{-\alpha}$$
Hence, choosing $a$ such that $m(\tau)= 1-e^{\frac{-1}{1-e^{-\alpha}}-2-\alpha}$ results in the violation of one of the Jury condition, and therefore is indicative of a  Neimark--Sacker bifurcation. 
\end{proof}

\section{Proof of Theorem~\ref{thm:GASNstarBH}}\label{A:thmGASNstarBH}


\begin{proof}
    Since $F_2(N_t,N_{t-\tau})$ is increasing in both of its variables and the only solutions to 
    $$r=F_2(r,r)\qquad \mbox{and}\qquad R=F_2(R,R)$$
    is $r,R\in \{0,N^*_+\}$ it suffices to show that there exists $a_r, b_r$ such that  $0<a_r<N^*_+< b_r$ and $\mathcal{I}=[a_r, b_r]^{\tau+1}$ is positively invariant and that solutions eventually enter $\mathcal{I}$. 
To show that there exists $a_r$, we note that 
$F_2(a,a)\geq a$ if and only if 
$$\frac{a}{1+\alpha+\beta a}+\frac{Dg(\tau)a}{D(1+D)^{\tau+1}+((1+D)^{\tau+1}-1)Cg(\tau)a}\geq a, $$
i.e., for $a>0$, 
$$\frac{1}{1+\alpha+\beta a}+\frac{Dg(\tau)}{D(1+D)^{\tau+1}+((1+D)^{\tau+1}-1)Cg(\tau)a}\geq 1, $$
which can be equivalently expressed as 
$q(a):=\gamma_0+\gamma_1a + \gamma_2a^2\leq 0,$
for 
\begin{align*}
    \gamma_0&=-Dg(\tau)(1+\alpha)+\alpha D(1+D)^{\tau+1},\\
    \gamma_1&=((1+D)^{\tau+1}-1)Cg(\tau) \alpha+
\beta  D\left((1+D)^{\tau+1}
- g(\tau)\right),\\
\gamma_2&=\beta ((1+D)^{\tau+1}-1)Cg(\tau)>0.
\end{align*}
By assumption \eqref{eq:condNstar}, 
$$\frac{\alpha}{1+\alpha} <\frac{g(\tau)}{(1+D)^{\tau+1}},$$
so that $$\alpha (1+D)^{\tau+1}<g(\tau)(1+\alpha),$$
and, therefore, 
$$\gamma_0=-Dg(\tau)(1+\alpha)+\alpha D(1+D)^{\tau+1}<0.$$

 Thus, there exists a positive root $a_\ell$, such that $q(a_\ell)=0$. Then,  $q(a)\leq 0$ for all $a\in (0, a_\ell)$ and, therefore, $F_2(a,a)\geq a_\ell$ for all $a\leq a_\ell$.  

To show that there exists $b_r>a_r$ such that $F_2(b_r, b_r)\leq b_r$, we note that $F_2(b,b)\leq b$ if and only if 
$$\frac{b}{1+\alpha+\beta b}+\frac{Dg(\tau)b}{D(1+D)^{\tau+1}+((1+D)^{\tau+1)-1)Cg(\tau)b}}\leq b. $$
Proceeding similarly as above, we can show that this is equivalent to 
$0\leq q(b)$. 
Again, since $\lim_{b\to \infty}q(b)=\infty$ and $q(0)<0$, 
$q(b)\geq 0$ for all $b\geq a_\ell$. Thus, we can choose $a_r\in (0, a_\ell)$ arbitrarily small and $b_r>a_\ell$ arbitrarily big.  This allows to choose $a_r<N^*_+<b_r$ such that $[a_r, b_r]$ is positively invariant. For a given solution, we choose $a_r=\min\{a_\ell, \frac{1}{2}N^*_+, \min_{i\in \{-\tau, \ldots, 0\}}N_i\}$ and $b_r=\max\{a_\ell, 2N^*_+, \max_{i\in \{-\tau, \ldots, 0\}}N_i\}$ to ensure that the solution remains in $[a_r, b_r]$ and, by Theorem~\ref{A:ThmLadas}, converges to $N^*_+$. This completes the proof. 
\end{proof}

\section{Proof of no period doubling bifurcation for "Ricker-Ricker'' model with even $\tau$}\label{A:thmRR_PD}
\begin{theorem}\label{thm:RR_PD}
     Consider \eqref{eqn:RickerRicker} with assumptions $a,b,K,\alpha,D>0$, and $\tau\geq0$, so $g(\tau)>0$. When $\tau$ is even, there is no $a>0$ where a period doubling bifurcation occurs for the trivial equilibrium.
\end{theorem}
\begin{proof}
   The Jacobian evaluated at the trivial equilibrium has the general format
\begin{equation*} 
 J=\begin{bmatrix}
e^{-\alpha} & 0 & 0 & \ldots & 0 & e^{-D}\\
 g(\tau)e^{-D} & 0 & 0 & \ldots & 0 & 0\\
0 & e^{-D} & 0 & \ldots & 0 & 0\\
\ldots & \ldots & \ldots & \vdots\vdots\vdots & \ldots & \ldots\\
0 & 0 & 0 & \ldots & e^{-D} & 0\\
\end{bmatrix}.
\end{equation*}

Let $\tau$ be even. For the sake of contradiction, suppose $\lambda=-1$ is an eigenvalue. Then, the first component of the corresponding eigenvector satisfies
$$v_1=-e^{-\alpha}v_1-(-1)^{\tau}e^{-D(\tau+1)}g(\tau)v_1,$$
i.e., $$\left\{1+e^{-\alpha}+e^{-D(\tau+1)}g(\tau)\right\}v_1=0,$$
which is only satisfied if $v_1=0$. In that case however, $v_i=0$ for all $i\in \{2,\ldots, \tau+1\}$, violating the property of an eigenvector. Hence, $\lambda=-1$ cannot be an eigenvalue. 
\end{proof}

\newpage

\section{Ricker-Ricker $C=\beta$} \label{A:cequalbeta}
\begin{figure}[H]
    \centering
    \includegraphics[width=0.5\linewidth]{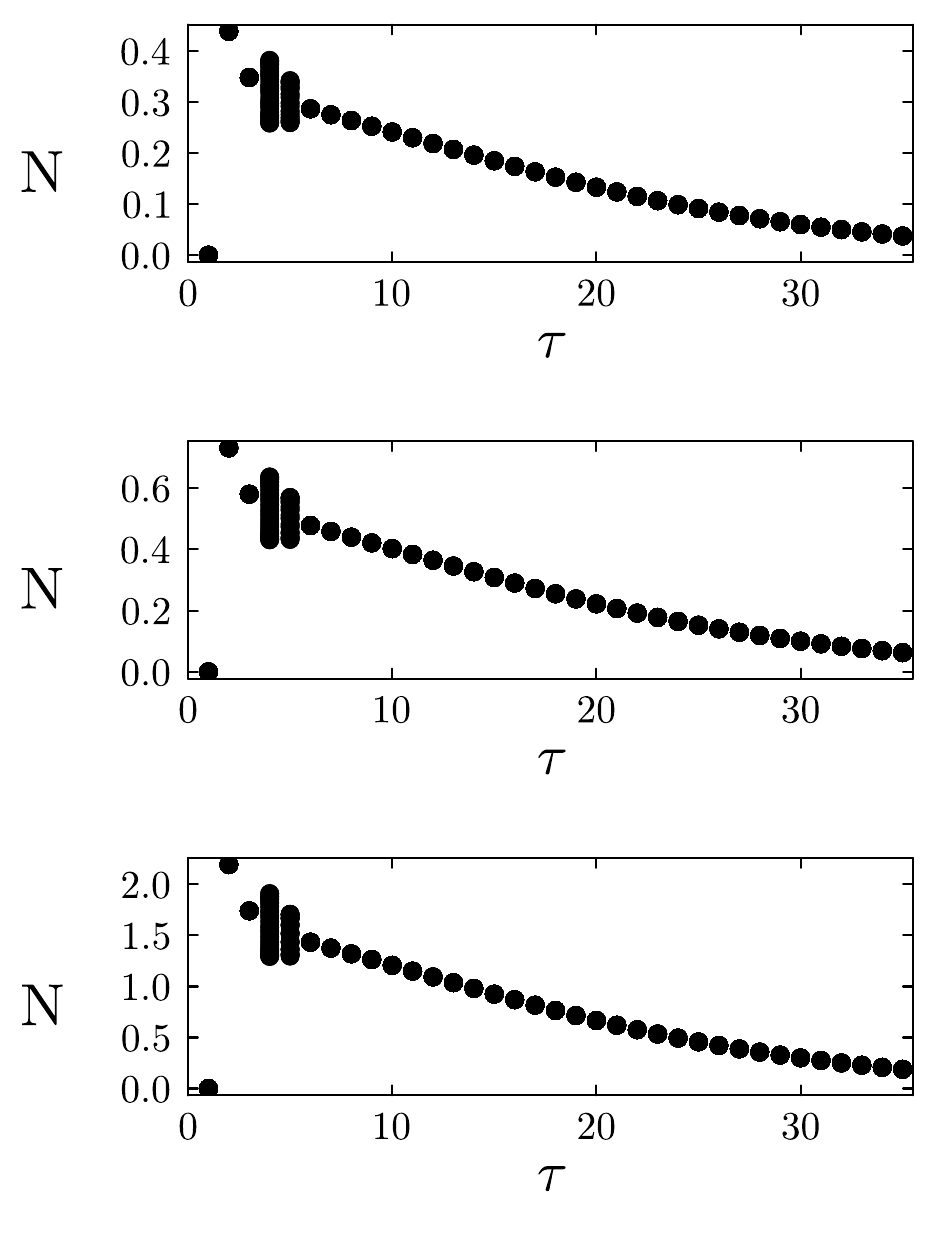}
    \caption{Top) Ricker-Ricker orbit diagram with $C=\beta=0.6$. Middle) Ricker-Ricker orbit diagram with $C=\beta=0.3$ Bottom) Ricker-Ricker orbit diagram with $C=\beta=0.1$. For all plots, the base parameters are $C=0.3$, $\beta=0.3$,  $D=0.1$, $\alpha=0.1$, $a=26.5$, $b=200.0$, $K=1.0$} 
    \label{fig:RR_Cequalbeta}
\end{figure}
\newpage

\bibliographystyle{unsrt}
\bibliography{Lib}

\end{document}